\documentclass[conference]{IEEEtran}
\usepackage{hyperref}
\usepackage[utf8]{inputenc}
\usepackage[utf8]{inputenc}
\usepackage{graphicx}
\usepackage{amsmath}
\usepackage{amsthm}
\usepackage{amssymb,stmaryrd,amsmath,amsfonts,rotating}
\usepackage{amsmath, amsthm, amssymb}
\usepackage{epstopdf}
\usepackage[mathscr]{eucal}  
\usepackage{float} 
\usepackage{multicol}
\usepackage{caption}
\usepackage{subcaption}
\usepackage{enumerate}
\usepackage{authblk}
\usepackage{multirow}
\usepackage{tikz}
\usetikzlibrary{decorations.pathreplacing}
\usepackage{multirow}

\usetikzlibrary{positioning}

\usepackage{tikz}
    \usetikzlibrary{shapes.geometric}   
\usepackage{rotating}

\usepackage{stmaryrd}
\usepackage{mathtools}
\usepackage{stfloats}
\usepackage{url}
\usepackage{amssymb}
\usepackage{mathabx}
\usepackage{algorithmicx}

\newcommand{\defn}{\stackrel{\Delta}{=}}

\newtheorem{proposition}{Proposition}

\newtheorem{theorem}{Theorem}
\newtheorem{lemma}{Lemma}

\newtheorem{Definition}{Definition}

\def\minus{\mathord{-}}
\def\plus{\mathord{+}}

\renewcommand{\theenumi}{\roman{enumi}}

\begin{document}
%


\title{Polar Codes with Higher-Order Memory}
\author{H\" useyin~Af\c ser}
\author{Hakan~Deli\c c}
\affil{Wireless Communications Laboratory, Department of Electrical and Electronics Engineering \\ Bo\u gazi\c ci University, Bebek 34342, Istanbul, Turkey \\
\{huseyin.afser,delic\}@boun.edu.tr}

\maketitle

\begin{abstract}
We introduce the design of a set of code sequences  $ \{ {\mathscr C}_{n}^{(m)} : n\geq 1, m \geq 1 \}$, 
with memory order $m$ and code-length $N=O(\phi^n)$, where $ \phi \in (1,2]$ is the largest real root of the polynomial equation $F(m,\rho)=\rho^m-\rho^{m-1}-1$ and $\phi$ is decreasing in $m$. $\{ {\mathscr C}_{n}^{(m)}\}$  is based on the channel polarization idea, where $ \{ {\mathscr C}_{n}^{(1)} \}$ coincides with the polar codes presented by Ar\i kan in \cite{Arikan} and can be encoded and decoded with complexity $O(N \log N)$. $ \{ {\mathscr C}_{n}^{(m)} \}$ achieves the symmetric capacity, $I(W)$, of an arbitrary binary-input, discrete-output memoryless channel, $W$, for any fixed $m$ and its encoding and decoding complexities decrease with growing $m$. We obtain an achievable bound on the probability of block-decoding error, $P_e$, of
$\{ {\mathscr C}_{n}^{(m)} \}$ and showed that $P_e = O (2^{-N^\beta} )$ is achievable for $\beta < \frac{\phi-1}{1+m(\phi-1)}$.

\end{abstract} 


\section{Introduction and Overview}

\let\thefootnote\relax\footnotetext{This work was supported by Bo\u gazi\c ci University Research Fund under Project
 11A02D10. H. Af\c{s}er was also supported by Aselsan Elektronik A.\c{S} }
Channel polarization \cite{Arikan} is a method to achieve the symmetric capacity, $I(W)$, of an arbitrary binary-input, discrete-output memoryless channel (B-DMC), $W$. By applying channel combining and splitting operations \cite{Arikan_2006}, one transforms
$N$ uses of $W$ into another set of synthesized binary-input channels. As $N$ increases, the symmetric capacities of the synthesized 
binary-input channels polarize as $I(W)$ fraction of them gets close to 1 and $1-I(W)$ fraction of them gets close to 0.  The resulting code sequences, called polar codes, have encoding and decoding complexities $O(N \log N)$, and their block error probabilities scale as $2^{-N^{\beta}}$ where $\beta < 1/2$ is the exponent of the code \cite{Arikan_rate}.

Let $W : {\cal X} \rightarrow {\cal Y}$ denote a B-DMC with binary-input $x \in {\cal X}=\{0,1\}$ and arbitrary discrete-output $y \in {\cal Y}$. 
Considering Ar\i kan's polar codes, let us write $W_n$ to denote the vector channel, $W_n : {\cal X}^N \rightarrow {\cal Y}^N$, $N=2^n$, $n \geq 1$, obtained at channel combining level $n$. The vector channel, $W_n$,  is obtained from $W_{n-1}$ in a recursive manner where one first injects an independent realization of $W_{n-1}$, denoted as $\hat{W}_{n-1}$, and then combines the input of $W_{n-1}$ and $\hat{W}_{n-1}$ to obtain $W_n$, where the recursion starts with $W_0=W$. The injection of $\hat{W}_{n-1}$, in a way, creates $N/2$ diversity paths for the $N/2$ inputs of $W_{n-1}$, and this allows polarization which one sees in the synthesized binary-input channels obtained by splitting $W_n$. Consequently, at each combining level 
the code-length doubles with respect to the previous step scaling as $N=2^n$.

With higher-order memory in channel polarization, let us write $N=N(n,m)$ to denote the code-length at channel combining level $n$ and memory parameter $m$, $ m\geq 1$, which we assume to be fixed. The vector channel, $W_n$, is obtained by combining the inputs of $W_{n-1}$ with $\hat{W}_{n-m}$,  where one chooses $W_0=W_{-1}=\ldots=W_{1-m}=W$ to initiate the recursion. The number of binary-inputs in $W_{n-1}$ and $\hat{W}_{n-m}$ are $N(n-1)$ and $N(n-m)$, respectively. In turn, with the controlled memory parameter, $m$, and at channel combining level $n$, one only injects $N(n-m)$  new diversity paths with $\hat{W}_{n-m}$, for the $N(n-1)$ inputs of $W_{n-1}$, to obtain $W_n$. Because $N(n-m)$ gets smaller compared to $N(n-1)$ as $m$ increases, it is possible to slow the speed at which one inject new channels to provide polarization. At first glance, it seems that increasing $m$ will decrease the polarization effect obtained after each combining and splitting stage, however it will also allow the code-length to increase less rapidly in $n$.  In order to see this consider the code-length obeying the recursion   
\begin{align}
N = N(n-1) + N(n-m), \quad n \geq 1, m \geq 1, \label{eqn:code_length_recursion}
\end{align}
with initial conditions
\begin{align}
N(0)  = N(-1)=\ldots=N(1-m)=1,   \quad m  \geq 1. \label{eqn:code_length_init}
\end{align}
As will be explained in the sequel, the code-length takes the form
\begin{align}
N  = O (\phi^n),  \quad n \geq 1 
\end{align}
where $\phi \in (1,2]$ is the largest real root of the $m$-th order polynomial equation
\begin{align}
F(m,\rho)=\rho^m-\rho^{m-1}-1,
\end{align}
and $\phi$ decreases with increasing $m$. Therefore, if we increase $m$, it will take more channel combining and splitting stages to reach a pre-defined code-length, where the ratio of injected diversity paths to existing paths in each combining stage will also decrease. The aim of this paper is to understand the effects of this trade-off on the polarization performance one can obtain at a fixed code-length $N$.

The original construction of polar codes by Ar\i kan is closely related to the recursive construction of Reed-Muller codes based on the $2 \times 2$ kernel $\textbf{F}_2= \left[\begin{smallmatrix} 1 & 0 \\ 1 & 1\end{smallmatrix} \right]$. For these codes the encoding matrix, $\textbf{G}_N$, is of the form  $\textbf{G}_N =\textbf{F}_2^{\otimes n}$, where $\otimes$ denotes the Kronecker power, suitably defined in \cite{Arikan}. In \cite{Korada_exponent} Korada et al. generalize the channel polarization idea where $\ell \geq 2$ independent uses of $W_{n-1}$ are arbitrarily combined to obtain $W_{n}$ and code-length scales as $N=\ell^n$. Although the channel combining mechanism is generalized to combining arbitrary numbers of $W_{n-1}$ to obtain $W_n$, this setup has also first order memory in the channel combining. The authors express the combining mechanism by an $\ell \times \ell$ polarization kernel $\textbf{K}_\ell$. With an arbitrary $\textbf{K}_\ell$, the encoding matrix takes the form $\textbf{G}_N=\textbf{K}_\ell^{\otimes n}$. The asymptotic polarization performance is characterized by the distance properties of the rows of $\textbf{K}_\ell$. The encoding and decoding complexities of these polar codes increases with $l$ scaling as $O(l N \log N)$ and $O( \frac{2^l}{l} N \log N)$, respectively. Our work differs from \cite{Korada_exponent} in the sense that by introducing higher-order memory we modify the channel combining process. Moreover the encoding matrix of polar codes with memory $m>1$ can not be obtained by applying Kronecker power to an arbitrary polarization kernel. As a result, one needs new mathematical tools to investigate $\beta$. 

The contributions of this paper are as follows: \textit{i)} We present a novel polar code family, $\{ {\mathscr C}_{n}^{(m)}:  n\geq , m\geq 1\}$, with code-length $N=O(\phi^n)$, $\phi \in (1,2]$, and arbitrary but fixed memory parameter $m$. We show that $\{ {\mathscr C}_{n}^{(m)} \}$ achieves the symmetric capacity of arbitrary BDMCs for any choice of $m$ which complements Ar\i kan's conjecture that channel polarization is in fact a general phenomenon. \textit{ii)} By developing a new mathematical framework, we obtain an asymptotic bound on the achievable exponent, $\beta$, of $\{ {\mathscr C}_{n}^{(m)} \}$. \textit{iii)} We show that the encoding and decoding complexities of  $\{ {\mathscr C}_{n}^{(m)} \}$ decrease with increasing $m$. $\{ {\mathscr C}_{n}^{(m)} \}$ is the first example of a polar code family that has  lower complexity compared to the original codes presented by Ar\i kan. 


The outline of the paper is a as follows. Section~\ref{sect:2} provides the necessary material for the analysis in the sequel. In Section~\ref{sect:3} we explain the design, encoding and the decoding of $\{ {\mathscr C}_{n}^{(m)}  \}$. In Section~\ref{sect:4} we develop a probabilistic framework to investigate $\{ {\mathscr C}_{n}^{(m)} \}$. After showing that $\{ {\mathscr C}_{n}^{(m)}  \}$ achieves the symmetric capacity of arbitrary B-DMCs we obtain an achievable bound on its block-decoding error probability. In Section~\ref{sect:5} we analyze impact of higher-order memory on the encoding and decoding complexities of $\{ {\mathscr C}_{n}^{(m)}  \}$. Section~\ref{sect:6} concludes the paper and provides some future research directions.  

\textbf{Notation:}  We use uppercase letter $A,B$ for random variables and lower cases $a,b$ for their realizations taking values
from sets $\cal A$, $\cal B$, where the sets have sizes $|\cal A|$ and $|\cal B|$ respectively. $\Pr(a)$ denotes the probability of the event $A=a$.  We write $\textbf{a}_n=(a_1,a_2,\ldots,a_n)$ to denote a vector and $(\textbf{a}_n,\textbf{b}_n)$ to denote the concatenation of $\textbf{a}_n$
and $\textbf{b}_n$. We use standard Landau notation $o(n), O(N)$ to denote the limiting values of functions.  \textbf{Note:} Proofs, unless stated otherwise, are provided in the Appendix.

\section{Preliminaries}\label{sect:2}
Let $W(y|x)$, $x \in {\cal X}$, $y \in {\cal Y}$ denote the transition probabilities of $W$. Throughout the paper we assume that $x$ is uniformly distributed in ${\cal X}$, and use base-2 logarithm. The symmetric capacity, $I(W)$, of $W$ is 
\begin{align}
 I(W) \defn \sum_{y \in \cal{Y}} \sum_{x \in \cal{X}} \frac{1}{2} W(y|x) \log \frac{W(y|x)}{\frac{1}{2}W(y|0) + \frac{1}{2}W(y|1) }. \label{eqn:I_def}
\end{align}
The Bhattacharyya parameter, $Z(W)$, of $W$ provides an upper bound on the probability of error for maximum likelihood (ML) decoding over $W$ and is defined as
\begin{align}
 Z(W) \defn \sum_{y \in \cal{Y}} \sqrt{W(y|0)W(y|1)}.\label{eqn:Z_def}
\end{align}
The symmetric cut-off rate, $J(W)$, of $W$ is \cite{Arikan} 
\begin{align}
 J(W) \defn \log \frac{2}{1+Z(W)}. \label{eqn:R_0_def}
\end{align}
As Ar\i kan shows in \cite[Prop. 1]{Arikan} $Z(W)=1$ implies $I(W)=0$ and $Z(W)=0$ implies $I(W)=1$. By using this fact and from \eqref{eqn:R_0_def} we see that if $J(W)=0$ then $I(W)=0$ holds and $J(W)=1$ indicates $I(W)=1$.

Let $W'$ and $W''$ be two B-DMCs with inputs $x_1 ,x_2\in {\cal X}$ and outputs $y_1 \in {\cal Y}_1$ and $y_2 \in {\cal Y}_2$, respectively. Channel polarization is based on a single-step channel transformation where one first combines the inputs of $W'$ and $W''$ to obtain a vector channel
\begin{align}
W(y_1,y_2|x_1,x_2)=W'(y_1|x_1 \oplus x_2) W''(y_2|x_2).
\end{align}
Next, by choosing a channel ordering, one splits the vector channel to obtain two new binary-input channels, $W^-:  {\cal X} \rightarrow {\cal Y}_1 \times {\cal Y}_2$
and $W^+ :{\cal X} \rightarrow  {\cal X} \times {\cal Y}_1 \times {\cal Y}_2$, with transition probabilities
\begin{align}
W^-(y_1,y_2|x_1) &=\sum_{x_2} \frac{1}{2}W'(y_1|x_1 \oplus x_2) W''(y_2|x_2), \label{eqn:W_minus}  \\
W^+(y_1,y_2,x_1|x_2) &=\frac{1}{2}W'(y_1|x_1 \oplus x_2) W''(y_2|x_2), \label{eqn:W_plus}
\end{align}
We use the following short-hand notations for the transforms in \eqref{eqn:W_minus} and \eqref{eqn:W_plus}, respectively. 
\begin{align}
W^-=W' \boxminus W'', \\ 
W^+=W' \boxplus W''.
\end{align}
The polarization transforms preserve the symmetric capacity as
\begin{align}
I(W^-) + I(W^+) = I(W') + I(W''), \label{eqn:I_sum}
\end{align}
and they help polarization by creating disparities in $I(W^+)$ and $I(W^-)$ such that
\begin{gather} 
I(W^+) \geq \max \{ I(W') ,I(W'') \}, \label{eqn:I_plus} \\
I(W^-)  \leq   \min \{ I(W') ,I(W'') \}, \label{eqn:I_minus} 
\end{gather}
where the above inequalities are strict as long as $I(W') \in (0,1)$ and $I(W'') \in (0,1)$. 
This polarization effect quantitatively observed in the  Bhattacharyya parameters as they take the form
\begin{gather} 
Z(W^+) = Z(W' )Z(W''),\label{eqn:Z_plus} \\
Z(W^-) \leq Z(W')+Z(W'')-Z(W')Z(W''), \label{eqn:Z_minus} 
\end{gather}
where the equality in \eqref{eqn:Z_minus} is achieved if $Z(W')  \in \{0,1\}$ or $Z(W'')  \in \{0,1\}$, or if $W'$ and $W''$ are binary erasure channels (BECs).

Equations \eqref{eqn:I_sum}-\eqref{eqn:Z_minus} are proved in \cite{Arikan} when $W'$ is identical to $W''$. Their generalizations for the case  $W'$ and $W''$ are different channels are straightforward and omitted. The proposition below  will be crucial in the sequel.
\begin{proposition}\label{prop:R_0}
\begin{gather}
\nonumber
J(W^-)+J(W^+) \geq J(W') + J(W''),
\end{gather}
where equality is achieved only if $J(W') \in \{0,1\}$ or $J(W'') \in \{0,1\}$.
\end{proposition}  
The above proposition indicates that one can obtain coding gain by applying channel combining and splitting operations as long as the symmetric cut-off rate of $W'$ and $W''$ is in $(0,1)$, where the coding gain manifests itself as an increase in the sum cut-off rate of channels $W^-$ and $W^-$ compared to $W'$ and $W^+$. In this paper we use the parameters $J(W)$  and $I(W)$ together to show that $\{ {\mathscr C}_{n}^{(m)}\}$ achieves $I(W)$ of an arbitrary $W$, whereas the parameter $Z(W)$ will be used to characterize polarization performance of $\{ {\mathscr C}_{n}^{(m)}\}$.

\section{Polarization with Higher-Order Memory}\label{sect:3}


We develop a method to design a family of code sequences $\{ {\mathscr C}_{n}^{(m)} ; n\geq 1, m \geq 1\}$ with code-length $N=N(n,m)=O(\phi^n)$, $ \phi \in (1,2]$, and fixed memory order $m$. $\{ {\mathscr C}_{n}^{(m)} \}$ is based on the channel polarization idea of Ar\i kan in \cite{Arikan}. This section is devoted to explaining the design, encoding and decoding of $\{ {\mathscr C}_{n}^{(m)} \}$, while preparing some grounds for investigating its characteristics in the following sections.

\subsection{Channel Combining}
Consider an arbitrary B-DMC, $W$, where its $N$ independent uses take the form $W(\textbf{y}_N|\textbf{x}_N)=\prod_{i=1}^{N} W(y_i|x_i)$, $\textbf{x}_N \in {\cal X}^{N}$, $\textbf{y}_N \in {\cal Y}^{N}$.
Let $\textbf{u}_{N} \in {\cal X}^{N}$ be the binary information vector that needs to be transmitted over $N$ uses of $W$. Channel combining
phase creates a vector channel $W_n : {\cal X}^{N} \rightarrow {\cal Y}^{N}$ of the form
\begin{align*}
W_n(\textbf{y}_N| \textbf{u}_N) = \prod_{i=1}^{N}W(y_i|x_i),
\end{align*}
where $\textbf{x}_N= \textbf{u}_N\textbf{G}_N$. $\textbf{G}_N$ is an $N \times N$ encoding matrix where encoding takes place in GF($2$).

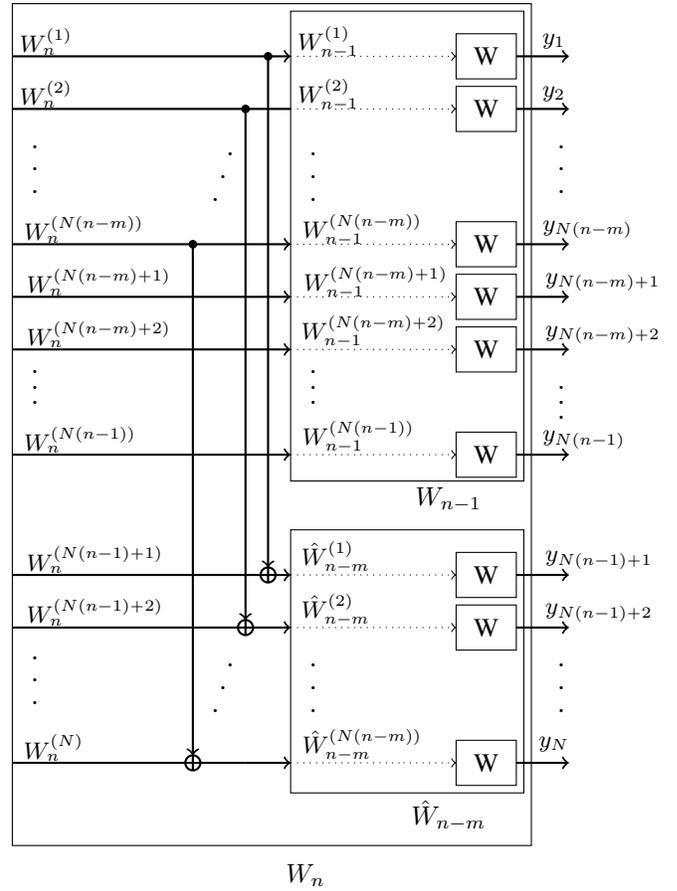
\begin{figure}[th]
\begin{center}

\begin{tikzpicture}[scale=1]


\draw [] (-1.1,0.1) rectangle (2,-6.15);
\node [] at (1,-6.4) {$W_{n-1}$};

\draw [thick] [->] (-4.8,-0.5) -- (-1.1,-0.5);

\draw [dotted] [-] (-1.1,-0.5) -- (0.4,-0.5);
\draw [] (1.1,-0.2) rectangle (1.9,-0.8);
\node [] at (1.5,-0.5) {W};
\draw [dotted] [->] (0.4,-0.5) -- (1.1,-0.5);

\draw [thick] [-] (-4.8,-1.2) -- (-1.1,-1.2);

\draw [dotted] [-] (-1.1,-1.2) -- (0.4,-1.2);
\draw [] (1.1,-0.9) rectangle (1.9,-1.5);
\node [] at (1.5,-1.2) {W};
\draw [dotted] [->] (0.4,-1.2) -- (1.1,-1.2);

\draw [thick] [->] (-4.8,-3) -- (-1.1,-3);
\draw [dotted] [-] (-1.1,-3) -- (0.4,-3);
\draw [] (1.1,-2.7) rectangle (1.9,-3.3);
\node [] at (1.5,-3) {W};
\draw [dotted] [->] (0.4,-3) -- (1.1,-3);

\draw [thick] [->] (-4.8,-3.7) -- (-1.1,-3.7);
\draw [dotted] [-] (-1.1,-3.7) -- (0.4,-3.7);
\draw [] (1.1,-3.4) rectangle (1.9,-4);
\node [] at (1.5,-3.7) {W};
\draw [dotted] [->] (0.4,-3.7) -- (1.1,-3.7);

\draw [thick] [->] (-4.8,-4.4) -- (-1.1,-4.4);
\draw [dotted] [-] (-1.1,-4.4) -- (0.4,-4.4);
\draw [] (1.1,-4.1) rectangle (1.9,-4.7);
\node [] at (1.5,-4.4) {W};
\draw [dotted] [->] (0.4,-4.4) -- (1.1,-4.4);

\draw [thick] [->] (-4.8,-5.8) -- (-1.1,-5.8);
\draw [dotted] [-] (-1.1,-5.8) -- (0.4,-5.8);
\draw [] (1.1,-5.5) rectangle (1.9,-6.1);
\node [] at (1.5,-5.8) {W};
\draw [dotted] [->] (0.4,-5.8) -- (1.1,-5.8);

\draw [] (-1.1,-6.8) rectangle (2,-10.3);
\node [] at (1,-10.6) {$\hat{W}_{n-m}$};

\draw[thick](-1.4,-7.4) circle [radius=0.1];
\draw [thick] [-] (-1.4,-7.3) -- (-1.4,-7.5);
\draw [thick] [->] (-4.8,-7.4) -- (-1.1,-7.4);

\draw[thick](-1.7,-8.1) circle [radius=0.1];
\draw [thick] [-] (-1.7,-8.2) -- (-1.7,-8.0);
\draw [thick] [->] (-1.2,-8.1) -- (-1.1,-8.1);
\draw [thick] [-] (-4.8,-8.1) -- (-1.2,-8.1);

\draw[thick](-2.4,-9.9) circle [radius=0.1];
\draw [thick] [-] (-2.4,-9.8) -- (-2.4,-10);
\draw [thick] [->] (-1.7,-9.9) -- (-1.1,-9.9);

\draw [thick] [-] (-4.8,-9.9) -- (-1.7,-9.9);

\node [] at (-1.9,-1.8) {\large{$.$}};
\node [] at (-2,-2.1) {\large{$.$}};
\node [] at (-2.1,-2.4) {\large{$.$}};

\node [] at (-4.5,-1.7) {\large{$.$}};
\node [] at (-4.5,-2) {\large{$.$}};
\node [] at (-4.5,-2.3) {\large{$.$}};

\node [] at ( -0.8,-1.8) {\large{$.$}};
\node [] at ( -0.8,-2.1) {\large{$.$}};
\node [] at ( -0.8,-2.4) {\large{$.$}};

\node [] at ( 2.5,-1.7) {\large{$.$}};
\node [] at ( 2.5,-2) {\large{$.$}};
\node [] at ( 2.5,-2.3) {\large{$.$}};


\node [] at (-4.5,-4.7) {\large{$.$}};
\node [] at (-4.5,-4.9) {\large{$.$}};
\node [] at (-4.5,-5.1) {\large{$.$}};

\node [] at (-0.8,-4.7) {\large{$.$}};
\node [] at (-0.8,-4.9) {\large{$.$}};
\node [] at (-0.8,-5.1) {\large{$.$}};

\node [] at (2.5,-4.9) {\large{$.$}};
\node [] at (2.5,-5.1) {\large{$.$}};
\node [] at (2.5,-5.3) {\large{$.$}};


\node [] at (-4.5,-8.5) {\large{$.$}};
\node [] at (-4.5,-8.8) {\large{$.$}};
\node [] at (-4.5,-9.1) {\large{$.$}};

\node [] at (-1.9,-8.6) {\large{$.$}};
\node [] at (-2,-8.9) {\large{$.$}};
\node [] at (-2.1,-9.2) {\large{$.$}};

\node [] at (-0.8,-8.6) {\large{$.$}};
\node [] at (-0.8,-8.9) {\large{$.$}};
\node [] at (-0.8,-9.2) {\large{$.$}};

\node [] at (2.5,-8.6) {\large{$.$}};
\node [] at (2.5,-8.9) {\large{$.$}};
\node [] at (2.5,-9.2) {\large{$.$}};

\draw [thick] [->] (-1.4,-0.5) -- (-1.4,-7.3);
\draw [thick] [->] (-1.7,-1.2) -- (-1.7,-8);
\draw [thick] [->] (-2.4,-3) -- (-2.4,-9.8);

\draw[fill](-1.4,-0.5) circle [radius=0.05];
\draw[fill](-1.7,-1.2) circle [radius=0.05];
\draw[fill](-2.4,-3) circle [radius=0.05];

\draw [thick] [->] (1.9,-0.5) -- (2.6,-0.5);
\draw [thick] [->] (1.9,-1.2) -- (2.6,-1.2);
\draw [thick] [->] (1.9,-3) -- (2.6,-3);
\draw [thick] [->] (1.9,-3.7) -- (2.6,-3.7);
\draw [thick] [->] (1.9,-4.4) -- (2.6,-4.4);
\draw [thick] [->] (1.9,-5.8) -- (2.6,-5.8);

\draw [thick] [->] (1.9,-7.4) -- (2.6,-7.4);
\draw [dotted] [-] (-1.1,-7.4) -- (0.4,-7.4);
\draw [] (1.1,-7.1) rectangle (1.9,-7.7);
\node [] at (1.5,-7.4) {W};
\draw [dotted] [->] (0.4,-7.4) -- (1.1,-7.4);

\draw [thick] [->] (1.9,-8.1) -- (2.6,-8.1);
\draw [dotted] [-] (-1.1,-8.1) -- (0.4,-8.1);
\draw [] (1.1,-7.8) rectangle (1.9,-8.4);
\node [] at (1.5,-8.1) {W};
\draw [dotted] [->] (0.4,-8.1) -- (1.1,-8.1);

\draw [thick] [->] (1.9,-9.9) -- (2.6,-9.9);
\draw [dotted] [-] (-1.1,-9.9) -- (0.4,-9.9);
\draw [] (1.1,-9.6) rectangle (1.9,-10.2);
\node [] at (1.5,-9.9) {W};
\draw [dotted] [->] (0.4,-9.9) -- (1.1,-9.9);

%

\node [above] at (-4.35,-0.6) {\small{$W_{n}^{(1)}$}};
\node [above] at (-4.35,-1.3) {\small{$W_{n}^{(2)}$}};

\node [above] at (-3.85,-3.1) {\small{$W_{n}^{(N(n-m))}$}};

\node [above] at (-3.65,-3.8) {\small{$W_{n}^{(N(n-m)+1)}$}};
\node [above] at (-3.65,-4.5) {\small{$W_{n}^{(N(n-m)+2)}$}};
\node [above] at (-3.9,-5.9) {\small{$W_{n}^{(N(n-1))}$}};

\node [above] at (-3.7,-7.5) {\small{$W_{n}^{(N(n-1)+1)}$}};
\node [above] at (-3.7,-8.2) {\small{$W_{n}^{(N(n-1)+2)}$}};
\node [above] at (-4.25,-10) {\small{$W_{n}^{(N)}$}};

\node [above] at (-0.6,-0.65) {\small{$W_{{n-1}}^{(1)}$}};
\node [above] at (-0.6,-1.35) {\small{$W_{{n-1}}^{(2)}$}};

\node [above] at (-0.15,-3.15) {\small{$W_{{n-1}}^{(N(n-m))}$}};
\node [above] at (0.03,-3.85) {\small{$W_{{n-1}}^{(N(n-m)+1)}$}};
\node [above] at (0,-4.5){\small{$W_{{n-1}}^{(N(n-m)+2)}$}};

\node [above] at (-0.2,-5.9) {\small{$W_{{n-1}}^{(N(n-1))}$}};

\node [above] at (-0.5,-7.5) {\small{$\hat{W}_{{n-m}}^{(1)}$}};
\node [above] at (-0.5,-8.2) {\small{$\hat{W}_{{n-m}}^{(2)}$}};

\node [above] at (-0.15,-10) {\small{$\hat{W}_{n-m}^{(N(n-m))}$}};

\node [] at (2.4,-0.3) { {\small{$y_1$}}};
\node [] at (2.4,-1) { {\small{$y_2$}}};
\node [] at (2.85,-2.8) { {\small{$y_{N(n-m)}$}}};

\node [] at (3.05,-3.5) { {\small{$y_{N(n-m)+1}$}}};
\node [] at (3.05,-4.2) { {\small{$y_{N(n-m)+2}$}}};
\node [] at (2.8,-5.6) { {\small{$y_{N(n-1)}$}}};

\node [] at (3,-7.2) { {\small{$y_{N(n-1)+1}$}}};
\node [] at (3,-7.9) { {\small{$y_{N(n-1)+2}$}}};
\node [] at (2.4,-9.65) { {\small{$y_{N}$}}};


\draw [] (-4.8,0.2) rectangle (2.1,-11);
\node [] at (-0.9,-11.4) {$W_n$};

\end{tikzpicture}
\end{center}
\caption{
Recursive construction of  the vector channel $W_n$ from $W_{n-1}$ and  $\hat{W}_{n-m}$, where $W_n^{(i)}$, $i \in {\mathbb N}_n$, denotes the binary-input channels in $W_n$. The arrows on the left show the directions of flow for the binary-inputs of $W_n^{(i)}$ and $\oplus$ is the XOR operation. The arrows on the right show the outputs of successive uses of $W$. The XOR operations that take place on the dotted arrows within  $W_{n-1}$ and $\hat{W}_{n-1}$ are not shown as they obey the same recursion.
}
\label{fig:S_channel_combining}
\end{figure}

Let ${\mathbb N}_{n}= \{1,2,\ldots,N\}$, $N=O(\phi^n)$, denote the set of the indices at the channel combining level $n$. There are $N$ binary-input channels in $W_n$ to transmit information. We index those  channels as $W_n^{(i)}$, $ i \in \mathbb{N}_{n}$, and demonstrate the channel combining operations in Fig~\ref{fig:S_channel_combining}.  Inspecting this figure observe that we index the topmost binary-input channel of $W_n$ as $W_n^{(1)}$ and index $i$ of $W_n^{(i)}$ increases as one move downwards. The vector channel $W_n$ is obtained by combining $W_{n-1}$ with $\hat{W}_{n-m}$. To accomplish this combining we apply XOR operations on the binary-inputs of $W_n$ and transmit the resultant bits through the inputs of $W_{n-1}$ and  $\hat{W}_{n-m}$. By continuing the same recursion within $W_{n-1}$ and $\hat{W}_{n-m}$, the encoded bits are transmitted through independent uses of $W$ channels because we start the combining recursion by choosing $W_0=W_{-1}=\ldots=W_{1-m}=W$. If we use  the binary-input channels $W_n^{(1)},W_n^{(2)},\ldots,W_n^{(N)}$ to transmit the symbols $u_1,u_2,\ldots,u_N$, respectively, the encoding matrix $\textbf{G}_N$ can be expressed as
\begin{align}\label{eqn:encoding_matrix}
\textbf{G}_N=\left[
\begin{array}{cc|c}
	\multicolumn{2}{c|}{\multirow{2}{*}{$ \textbf{G}_{N(n-1)} $}} &   \textbf{G}_{N(n-m)}\\ 
	\multicolumn{2}{c|}{} & \textbf{0}_2\\ \hline
	\multicolumn{2}{c|}{\textbf{0}_1} & \textbf{G}_{N(n-m)}\\
\end{array} \right], \quad n \geq 1
\end{align}
where $\textbf{G}_{N(0)}=\textbf{G}_{N(-1)}=\ldots =\textbf{G}_{N(1-m)}=[1]$, and $\textbf{0}_1$ and $\textbf{0}_2$ are $N(n-m) \times N(n-1)$ and $(N(n-1)-N(n-m)) \times N(n-m)$ all zero matrices, respectively. Observe that when $m=1$, $\textbf{0}_2$ matrix vanishes and $\textbf{G}_N$ can be represented as $\textbf{G}_n = ({\textbf{F}_2^\intercal})^{\otimes n}$, where  $\textbf{F}_2=\left[\begin{smallmatrix} 1 & 0 \\ 1 & 1\end{smallmatrix} \right]$ is the Kernel used by Ar\i kan in \cite{Arikan}. However, when $m>1$, $\textbf{G}_N$ can not be represented via Kronecker power.

\subsection{Channel Ordering}
After performing channel combining operation we have to define an order to split the vector $W_n: {\cal X}^N \rightarrow {\cal Y}^N$ and obtain $N$
binary-input channels. This ordering is carried out with the help of a permutation  $\pi_n : {\mathbb N}_n \rightarrow {\mathbb N}_n$. The $W_n^{(i)}$ channels in $W_n$ are split in increasing $\pi_n(i)$ values (from $1$ to $N$) so that each $W_n^{(i)}$ channel is of the form $W_n^{(i)} : {\cal X} \rightarrow {\cal Y}^{N} \times {\cal X}^{\pi(i)-1}$. In order to explain this operation we  associate a unique state vector  $\textbf{s}_n^{(i)}$ with each $W_n^{(i)}$ channel, which has the form
\begin{align*}
 \textbf{s}_n^{(i)}=(s_1^{(i)},s_2^{(i)},\ldots, s_n^{(i)}), 
\end{align*}
where 
\begin{align*}
\textbf{s}_k^{(i)} \in \{+,-,\bigstar\}, \quad k=1,2,\ldots,n 
\end{align*}
$s_k^{(i)}$ terms will be referred as a ``state" and we use $+,-,\bigstar$ symbols to track down the channel transformations that $W_n^{(i)}$ channels undergo as $n=1,2,\ldots$. States $+$, $-$ will correspond to the polarization transforms $\boxplus$ and $\boxminus$, as defined in \eqref{eqn:W_minus} and  \eqref{eqn:W_plus}, respectively; whereas state $\bigstar$ will correspond to a non-polarizing transform. We let 
\begin{align}
{\cal S}_n= \{ \textbf{s}_n^{(i)}: i  \in \mathbb{N}_n \}
\end{align}
to be the set of all possible state vectors at level $n$. Since each $\textbf{s}_n^{(i)} \in {\cal S}_n$  is unique (as we will show shortly) we have $|{\cal S}_n |=N$ and ${\cal S}_n \subset \{+,-,\bigstar\}^n$.
The vectors, $\textbf{s}_n^{(i)}\in {\cal S}_n$, are assigned recursively from $\textbf{s}_{n-1}^{(j)} \in {\cal S}_{n-1}$, with a state assigning procedure $\varphi_n:{\cal S}_{n-1} \rightarrow {\cal S}_{n}$.  The operation of $\varphi_n$ is explained in the following definition.
\begin{Definition}{\textbf{State Vector Assigning Procedure:}}
\label{defn:state_labeling}
Let  $\textbf{s}_{n-1}^{(j)} \in {\cal S}_{n-1}$  be the state vector of  $W_{n-1}^{(j)}$. The state vectors
$\textbf{s}_{n}^{(i)} \in  {\cal S}_{n}$, associated with $W_n^{(i)}$ take the form
\begin{align}
\begin{split}
\textbf{s}_n^{(j)} &= (\textbf{s}_{n-1}^{(j)},+),  \\
\textbf{s}_n^{(j+N(n-1))} &= (\textbf{s}_{n-1}^{(j)},-), 
\end{split}
 &j \in \mathbb{N}_{n-m}, \label{eqn:polarized_channels} \quad \quad  \hspace{6pt} \\
\textbf{s}_n^{(j)} &= (\textbf{s}_{n-1}^{(j)},\bigstar),     &j \in  {\mathbb N}_{n-1} \setminus {\mathbb N}_{n-m}.
\end{align}
\end{Definition}

\begin{figure}[t]
\begin{center}

\begin{tikzpicture}[scale=0.8]


\draw [] (-0.3,0.1) rectangle (2.1,-6.3);
\node [] at (0.9,-6.6) {${\varphi}_{n-1}$};

\draw [thick] [<-] (-2.4,-0.5) -- (-0.3,-0.5);

\draw [thick] [<-] (-2.4,-1.2) -- (-0.3,-1.2);
\draw [thick] [<-] (-2.4,-3) -- (-0.3,-3);
\draw [thick] [<-] (-2.4,-3.7) -- (-0.3,-3.7);
\draw [thick] [<-] (-2.4,-4.4) -- (-0.3,-4.4);
\draw [thick] [<-] (-2.4,-6) -- (-0.3,-6);


\draw [thick] [<-] (-2.4,-7.4) -- (-0.7,-7.4);


\draw [thick] [<-] (-2.4,-8.1) -- (-1,-8.1);


\draw [thick] [<-] (-2.4,-9.9) -- (-1.7,-9.9);

\node [] at (-1.2,-1.8) {\large{$.$}};
\node [] at (-1.3,-2.1) {\large{$.$}};
\node [] at (-1.4,-2.4) {\large{$.$}};

\node [] at (-3.4,-1.8) {\large{$.$}};
\node [] at (-3.4,-2.1) {\large{$.$}};
\node [] at (-3.4,-2.4) {\large{$.$}};

\node [] at ( 0.1,-1.8) {\large{$.$}};
\node [] at ( 0.1,-2.1) {\large{$.$}};
\node [] at ( 0.1,-2.4) {\large{$.$}};


\node [] at (-3.4,-4.8) {\large{$.$}};
\node [] at (-3.4,-5.1) {\large{$.$}};
\node [] at (-3.4,-5.4) {\large{$.$}};

\node [] at (0.1,-4.8) {\large{$.$}};
\node [] at (0.1,-5.1) {\large{$.$}};
\node [] at (0.1,-5.4) {\large{$.$}};


\node [] at (-3.4,-8.5) {\large{$.$}};
\node [] at (-3.4,-8.8) {\large{$.$}};
\node [] at (-3.4,-9.1) {\large{$.$}};

\draw [thick] [-] (-0.7,-0.5) -- (-0.7,-7.4);
\draw [thick] [-] (-1,-1.2) -- (-1,-8.1);
\draw [thick] [-] (-1.7,-3) -- (-1.7,-9.9);

\draw[fill](-0.7,-0.5) circle [radius=0.05];
\draw[fill](-1,-1.2) circle [radius=0.05];
\draw[fill](-1.7,-3) circle [radius=0.05];


%

\node [above] at (-3.3,-0.9) {\small{$  ( \textbf{s}_{n-1}^{(1)},+) $}};
\node [above] at (-3.3,-1.6) {\small{$  ( \textbf{s}_{n-1}^{(2)},+)$}};
\node [above] at (-3.8,-3.4) {\small{$ ( \textbf{s}_{n-1}^{(N(n-m))},+)$}};
\node [above] at (-4,-4.1)     {\small{$ ( \textbf{s}_{n-1}^{(N(n-m)+1)},\bigstar)$}};
\node [above] at (-4,-4.8){\small{$ ( \textbf{s}_{n-1}^{(N(n-m)+2)},\bigstar)$}};
\node [above] at (-3.8,-6.4){\small{$( \textbf{s}_{n-1}^{(N(n-1))},\bigstar)$}};

\node [above] at (-3.4,-7.8){\small{$( \textbf{s}_{n-1}^{(1)},-)$}};
\node [above] at (-3.4,-8.5){\small{$ ( \textbf{s}_{n-1}^{(2)},-)$}};
\node [above] at (-3.8,-10.3){\small{$ ( \textbf{s}_{n-1}^{(N(n-m))},-)$}};

\node [above] at (0.2,-0.9) {\small{$\textbf{s}_{n-1}^{(1)}  $}};
\node [above] at (0.2,-1.6) {\small{$\textbf{s}_{n-1}^{(2)}  $}};
\node [above] at (0.7,-3.4) {\small{$\textbf{s}_{n-1}^{(N(n-m)) }$}};
\node [above] at (0.9,-4.1)  {\small{$\textbf{s}_{n-1}^{(N(n-m)+1) }$}};
\node [above] at (0.9,-4.8) {\small{$\textbf{s}_{n-1}^{(N(n-m)+2) }$}};
\node [above] at (0.7,-6.4) {\small{$\textbf{s}_{n-1}^{(N(n-1)) }$}};


\draw [] (-2.4,0.2) rectangle (2.2,-10.3);
\node [] at (0.1,-10.6) {${\varphi}_n$};

\end{tikzpicture}
\end{center}
\caption
{ State labeling procedure  $ \varphi_n : {\cal S}_{n-1} \rightarrow {\cal S}_{n}$. State vectors $\textbf{s}_n^{(i)} \in {\cal S}_{n}$, are obtained by appending a new state $\{+,-,\bigstar\}$, to the vectors $\textbf{s}_{n-1}^{(j)}  \in {\cal S}_{n-1}$.
}
\label{fig:S_state_labeling}
\end{figure}
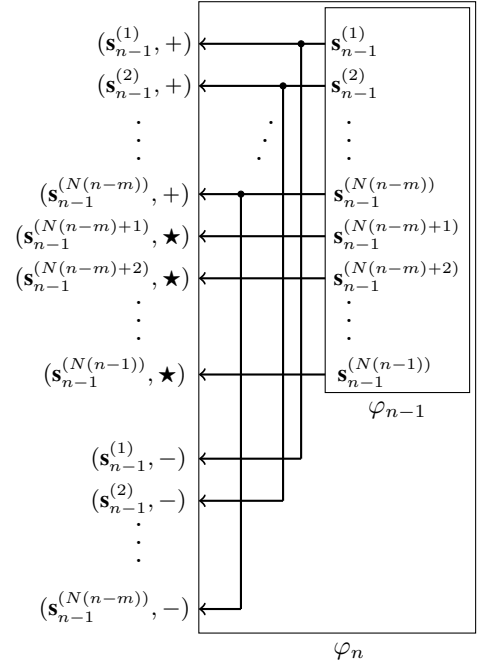

Investigating the above definition, as also demonstrated in Fig.~\ref{fig:S_state_labeling}, we observe that $\varphi_n$ appends a new state, $\{+,-,\bigstar\}$, to $\textbf{s}_{n-1}^{(j)} \in {\cal S}_{n-1}$ in order to construct $\textbf{s}_{n}^{(i)} \in {\cal S}_n$. For $j \in \mathbb{N}_{n-m}$, $\varphi_n$ appends $+$ and $-$ to $\textbf{s}_{n-1}^{(j)}$ to obtain  $\textbf{s}_n^{(j)}$ and $\textbf{s}_n^{(j+N_{n-1})}$, respectively. For $j \in  {\mathbb N}_{n-1} \setminus {\mathbb N}_{n-m}$,  $\varphi_n$ appends $\bigstar$ to $\textbf{s}_{n-1}^{(j)}$ in order to construct $\textbf{s}_n^{(j)}$. Because of the inherent memory in the combining procedure, it is difficult to obtain closed form expressions for $\textbf{s}_n^{(i)}$, for any $i$ and $m$. Nevertheless, with the above definition one can recursively obtain  $\textbf{s}_n^{(i)}$, by applying $\varphi_1,\varphi_2,\ldots,\varphi_n$. With the following proposition, we give the formal structure of the possible state vector, $\textbf{s}_n^{(i)}$, and thus the set $ {\cal S}_n$.

\begin{proposition}\label{prop:valid_states}
Let $\textbf{s}_n$, $\textbf{s}_n \in{\cal S}_n$, be a valid state vector one can obtain after applying $\varphi_1,\varphi_2,\ldots,\varphi_n$. Only the transitions between $s_{k}$ and $s_{k+1}$, $k=1,2,\ldots,n$, that are shown in the state transition diagram of Fig. \ref{fig:S_state_transition} are possible, where the imposed initial condition is $s_1 \in \{+,-\}$. 
\end{proposition}
The above proposition is a direct consequence of the channel combining and state vector assigning procedure, $\varphi_n$, and it can be verified
by induction through stages $\varphi_1,\varphi_2,\ldots,\varphi_n$. 

\begin{proposition}\label{prop:unique_states}
The state vector $\textbf{s}_n^{(i)} \in {\cal S}_n$, $ i \in \mathbb{N}_n$, assigned to each $W_n^{(i)} \in {\cal W}_n$  is unique.
\end{proposition}
The above proposition will be crucial for the ongoing analysis as it states that each $W_n^{(i)}$ is uniquely addressable by $\textbf{s}_n^{(i)}$. 
We will use this fact to obtain the ordering $\pi_n$. Before accomplishing this, we obtain binary vectors $\textbf{b}_n^{(i)}=(b_1^{(i)},b_2^{(i)}, \ldots,b_n^{(i)})$, $b_{k}^{(i)}  \in {\cal X}$, $k=1,2,\ldots,n$, from  $\textbf{s}_n^{(i)}$, which will allows us to sort and provide an order. The mapping between $\textbf{s}_n^{(i)}$ and  $\textbf{b}_n^{(i)}$ is obtained as
\begin{align}\label{eqn:binary_map}
b_{k}^{(i)} = \begin{cases}  0  \quad \text{if}  \quad s_k^{(i)} \in \{-,\bigstar \},      \\ 
			       1  \quad \text{if}  \quad s_k^{(i)} =  +, 	
	\end{cases}\quad k=1,2,\ldots,n.
\end{align}
We notice that although both $s_k^{(i)}=-$ and $s_k^{(i)}=\bigstar$ are mapped as $b_k^{(i)}=0$, the $\textbf{b}_n^{(i)}$ vectors will also be unique for each $i$ because every state $-$ in $\textbf{s}_n^{(i)}$ is followed by $m-1$ occurrences of state $\bigstar$, and the distinction between
different $\textbf{s}_n^{(i)}$ is hidden in the location of $+$ states in $\textbf{s}_n^{(i)}$.  The following definition uses this uniqueness property to obtain the ordering, $\pi_n$. It is an adaptation of the bit-reversed order of Ar\i kan in \cite{Arikan} to the proposed coding scheme.
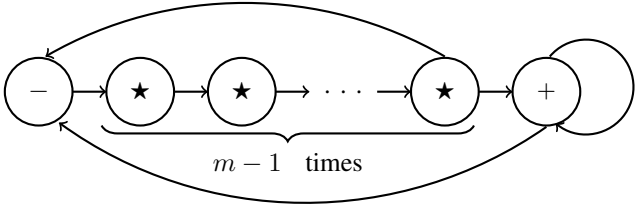
\begin{figure}[t]
\begin{center}

\begin{tikzpicture}[scale=0.9]

\draw [thick] (0,0) circle [radius=0.5];
\node [] at (0,0) {$-$};
\draw [thick] [->] (0.5,0) -- (1,0);

\draw [thick] (1.5,0) circle [radius=0.5];
\node [] at (1.5,0) {$\bigstar$};
\draw [thick] [->] (2,0) -- (2.5,0);

\draw [thick] (3,0) circle [radius=0.5];
\node [] at (3,0) {$\bigstar$};
\draw [thick] [->] (3.5,0) -- (4,0);


\node [] at (4.5,0) {.      .      .};

\draw [thick] [->] (5,0) -- (5.5,0);
\draw [thick] (6,0) circle [radius=0.5];
\node [] at (6,0) {$\bigstar$};
\draw [thick,->] (6,0.52) arc [radius=5.9, start angle=60, end angle=120];
\draw [thick] [->] (6.5,0) -- (7,0);
\draw [thick,->] (7.55,0.52) arc [radius=0.7, start angle=140, end angle=-130];

\draw [thick,->] (7.5,-0.52) arc [radius=6.2, start angle=-55, end angle=-126];

\draw [thick] (7.5,0) circle [radius=0.5];
\node [] at (7.5,0) {$+$};

\draw [thick] [decorate,decoration={brace,amplitude=6pt,mirror},xshift=-2pt,yshift=0pt]
(1, -0.5) -- (6.5,-0.5) node [black,midway,yshift=-0.5cm] {$m-1$ \hspace{2pt} \text{times}};

\end{tikzpicture}
\end{center}
\caption{
Possible state transitions observed between $s_{k}$ and $s_{k+1}$, $k=1,2,\ldots,n$.
}
\label{fig:S_state_transition}
\end{figure}
\begin{Definition}{\textbf{Bit-Reversed Order:}}\label{defn:channel_splitting_order}
Let $(\textbf{b}_n^{(i)})_2$ denote value of $\textbf{b}_n^{(i)}$ in Mod-2 as $(b_1^{(i)},b_{2}^{(i)},\ldots,b_n^{(i)})_2$ where $b_1^{(i)}$ is the most significant bit. The uniqueness of $\textbf{b}_n^{(i)}$ for each $i$ ensures the existence of  a permutation $\pi_n:  {\mathbb N}_n \rightarrow  {\mathbb N}_n$, so that for some $i,j \in {\mathbb N}_n$,  
we have $\pi_n(i) < \pi_n(j) $ if $(\textbf{b}_n^{(i)})_2 <(\textbf{b}_n^{(j)})_2$. 
\end{Definition}
Therefore the bit-reversed order $\pi_n$ is obtained in terms of increasing $(\textbf{b}_n^{(i)})_2$ values.

Notice that the binary input channels $\hat{W}_{n-m}^{(j)}$, $j \in \mathbb{N}_{n-m}$, of Fig.~\ref{fig:S_channel_combining} have no effect
in the recursive state assigning procedure, $\varphi_n$, and thus in the bit-reversed order. Their sole purpose is to provide auxiliary channels for the combining process. In fact, the $N(n-m)$ inputs of $\hat{W}_{n-m}$ can be combined with the $N(n-1)$ inputs of $\hat{W}_{n-1}$ in $\frac{N(n-1)!}{N(n-m)!}$ different ways. However, we deliberately align the inputs of $W_{n-1}$ and $\hat{W}_{n-m}$ so that the first $N(n-m)$ inputs of $W_{n-1}$ are combined, respectively, with the the first $N(n-m)$ inputs of $\hat{W}_{n-m}$  as shown in Fig.~\ref{fig:S_channel_combining}. This alignment in the combining process will be crucial  in the next section when we investigate the evolution of binary-input channels in a probabilistic setting, because the channel pairs, $W_{n-1}^{(j)}$ and $\hat{W}_{n-m}^{(j)}$, share the same state history as explained in the following proposition.

\begin{proposition}\label{prop:same_state_history}
Let $\textbf{s}_{n-1}^{(j)}=(s_1,s_2,\ldots,s_{n-1})\in {\cal S}_{n-1}$ be the state vector of $W_{n-1}^{(j)}$. Channel $\hat{W}_{(n-m)}^{(j)}$ shares the same state history with $W_{(n-1)}^{(j)}$, through combining stages $1,2,\ldots,n-m$, in the sense that its state vector is $\textbf{s}_{n-m}^{(j)}=(s_1,s_2,\ldots,s_{n-m}) \in {\cal S}_{n-m}$.
\end{proposition}

\subsection{Channel Splitting}
We assume a genie-aided decoding mechanism where the $W_n^{(i)}$ channels are decoded successively in increasing $\pi_n(i)$ values, from $1$ to $N$, and the genie provides the true values of already decoded bits. The decoder has no knowledge of the future bits that it will decode. With these assumptions $W_n^{(i)}$ is the effective bit-channel that this genie-aided decoder faces while trying to decode its next bit. 
Let us define $u_n^{(i)} \in {\cal X}$ as
\begin{align*}
u_n^{(i)}= \text{binary input of the channel } W_n^{(i)},
\end{align*}
and for $i,j \in {\mathbb N}_n$ let
\begin{align}
\begin{split}
 \textbf{u}_{n,b}^{(i)} \! \!& \defn  ( u_n^{(j)} : \pi_n{(j)} < \pi_n{(i)}   ),\\
  \textbf{u}_{n,a}^{(i)} \! \! & \defn   (  u_n^{(j)} : \pi_n{(j)} > \pi_n{(i)}  ). \label{eqn:U_b_and_U_a}
\end{split}
\end{align}
$ \textbf{u}_{n,b}^{(i)}$ and $ \textbf{u}_{n,a}^{(i)} $ are the information vectors that are decoded, by the genie-aided decoder, before and after $u_n^{(i)}$, respectively. The length of $\textbf{u}_{n,b}^{(i)} $ is $\pi_n{(i)}-1$ and the length of $\textbf{u}_{n,a}^{(i)} $ is $N-\pi_n{(i)}$ so that  $\textbf{u}_{n,b}^{(i)}  \in {\cal X}^{\pi_n{(i)}-1}$ and $\textbf{u}_{n,a}^{(i)} \in  {\cal X}^{N_n-\pi_n{(i)}}$. The following definition formalizes the transition probabilities of the $W_n^{(i)}$ channels.
\begin{align}
  W_n^{(i)} \! \defn  \sum_{   \textbf{u}_{n,a}^{(i)} }  \Pr \left( \textbf{y}_{N}, \textbf{u}_{n,a}^{(i)}  ,\textbf{u}_{n,b}^{(i)}|u_n^{(i)}  \right). \label{eqn:split_channel_defn}
\end{align}
The above definition indicates that $W_n^{(i)}$ is the posterior probability of an arbitrary B-DMC obtained at channel combining and splitting level $n$.
The genie-aided decoder has no knowledge of $\textbf{u}_{n,a}^{(i)}$, therefore
it averages the joint probability of all outputs and all inputs over $\textbf{u}_{n,a}^{(i)}$ and takes $\textbf{y}_{N}$ and $\textbf{u}_{n,b}^{(i)}$ as the effective output (observation) of the combined channels. Hence each $W_n^{(i)}$ has input $u_n^{(i)} \in {\cal X}$ and output $(\textbf{y}_{N},\textbf{u}_{n,b}^{(i)}) \in {\cal Y}^{N} \times {\cal X}^{\pi_n(i)-1}$.

\begin{proposition}\label{prop:S_split_channels}
The transition probabilities of $W_n^{(i)}$ channels take the following forms
\begin{align}
\begin{split}
W_n^{(j)} &= \hat{W}_{n-m}^{(j)} \boxplus  W_{n-1}^{(j)},  \\
W_n^{(j+{N_{n-1}})} &= \hat{W}_{n-m}^{(j)} \boxminus  W_{n-1}^{(j)}, 
\end{split}
\quad \quad\quad\quad j \in \mathbb{N}_{n-m}, \label{eqn:polarized_channels} \\
 W_n^{(j)} &=  \gamma(n)W_{n-1}^{(j)},   \quad \quad  \quad j \in  {\mathbb N}_{n-1} \setminus {\mathbb N}_{n-m},  \label{eqn:non_polarized_channels}
\end{align}
where $\gamma(n)= \Pr (y_{N(n-1)+1},y_{N(n-1)+2},\ldots,y_{N})$ and $W_{0}=W_{-1}=\ldots=W_{1-m}=W$.
\end{proposition}

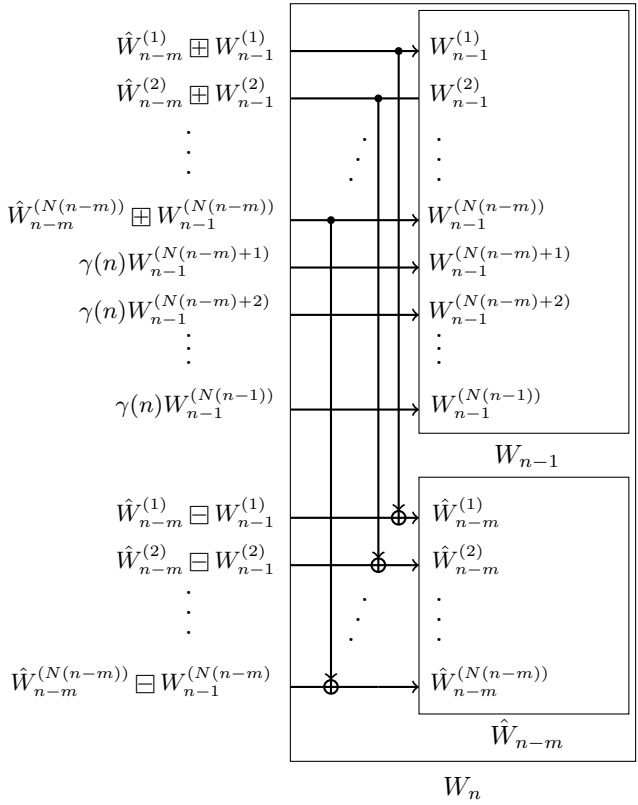
\begin{figure}[t]
\begin{center}

\begin{tikzpicture}[scale=0.9]


\draw [] (-1.1,0.1) rectangle (2,-6.15);
\node [] at (0.5,-6.5) {$W_{n-1}$};

\draw [thick] [->] (-3,-0.5) -- (-1.1,-0.5);
\draw [thick] [-] (-3,-1.2) -- (-1.1,-1.2);
\draw [thick] [->] (-3,-3) -- (-1.1,-3);
\draw [thick] [->] (-3,-3.7) -- (-1.1,-3.7);
\draw [thick] [->] (-3,-4.4) -- (-1.1,-4.4);
\draw [thick] [->] (-3,-5.8) -- (-1.1,-5.8);

\draw [] (-1.1,-6.8) rectangle (2,-10.3);
\node [] at (0.5,-10.6) {$\hat{W}_{n-m}$};

\draw[thick](-1.4,-7.4) circle [radius=0.1];
\draw [thick] [-] (-1.4,-7.3) -- (-1.4,-7.5);
\draw [thick] [->] (-3,-7.4) -- (-1.1,-7.4);

\draw[thick](-1.7,-8.1) circle [radius=0.1];
\draw [thick] [-] (-1.7,-8.2) -- (-1.7,-8.0);
\draw [thick] [->] (-1.2,-8.1) -- (-1.1,-8.1);
\draw [thick] [-] (-3,-8.1) -- (-1.2,-8.1);

\draw[thick](-2.4,-9.9) circle [radius=0.1];
\draw [thick] [-] (-2.4,-9.8) -- (-2.4,-10);
\draw [thick] [->] (-1.7,-9.9) -- (-1.1,-9.9);

\draw [thick] [-] (-3,-9.9) -- (-1.7,-9.9);

\node [] at (-1.9,-1.8) {\large{$.$}};
\node [] at (-2,-2.1) {\large{$.$}};
\node [] at (-2.1,-2.4) {\large{$.$}};

\node [] at (-4.5,-1.7) {\large{$.$}};
\node [] at (-4.5,-2) {\large{$.$}};
\node [] at (-4.5,-2.3) {\large{$.$}};

\node [] at ( -0.8,-1.8) {\large{$.$}};
\node [] at ( -0.8,-2.1) {\large{$.$}};
\node [] at ( -0.8,-2.4) {\large{$.$}};


\node [] at (-4.5,-4.7) {\large{$.$}};
\node [] at (-4.5,-4.9) {\large{$.$}};
\node [] at (-4.5,-5.1) {\large{$.$}};

\node [] at (-0.8,-4.7) {\large{$.$}};
\node [] at (-0.8,-4.9) {\large{$.$}};
\node [] at (-0.8,-5.1) {\large{$.$}};


\node [] at (-4.5,-8.5) {\large{$.$}};
\node [] at (-4.5,-8.8) {\large{$.$}};
\node [] at (-4.5,-9.1) {\large{$.$}};

\node [] at (-1.9,-8.6) {\large{$.$}};
\node [] at (-2,-8.9) {\large{$.$}};
\node [] at (-2.1,-9.2) {\large{$.$}};

\node [] at (-0.8,-8.6) {\large{$.$}};
\node [] at (-0.8,-8.9) {\large{$.$}};
\node [] at (-0.8,-9.2) {\large{$.$}};

\draw [thick] [->] (-1.4,-0.5) -- (-1.4,-7.3);
\draw [thick] [->] (-1.7,-1.2) -- (-1.7,-8);
\draw [thick] [->] (-2.4,-3) -- (-2.4,-9.8);

\draw[fill](-1.4,-0.5) circle [radius=0.05];
\draw[fill](-1.7,-1.2) circle [radius=0.05];
\draw[fill](-2.4,-3) circle [radius=0.05];


%

\node [above] at (-4.4,-0.8) {\small{$\hat{W}_{n-m}^{(1)} \boxplus W_{n-1}^{(1)}  $}};
\node [above] at (-4.4,-1.5) {\small{$\hat{W}_{n-m}^{(2)} \boxplus W_{n-1}^{(2)}  $}};

\node [above] at (-5.2,-3.3) {\small{$\hat{W}_{n-m}^{(N(n-m))} \boxplus W_{n-1}^{(N(n-m))} $}};

\node [above] at (-4.7,-4) {\small{$\gamma(n)  W_{n-1}^{(N(n-m)+1)} $}};
\node [above] at (-4.7,-4.7) {\small{$\gamma(n) W_{n-1}^{(N(n-m)+2)}   $}};
\node [above] at (-4.4,-6.1) {\small{$\gamma(n) W_{n-1}^{(N(n-1))}   $}};

\node [above] at (-4.4,-7.7) {\small{$\hat{W}_{n-m}^{(1)}  \boxminus W_{n-1}^{(1)} $}};
\node [above] at (-4.4,-8.4) {\small{$\hat{W}_{n-m}^{(2)} \boxminus W_{n-1}^{(2)} $}};
\node [above] at (-5.2,-10.2) {\small{$\hat{W}_{n-m}^{(N(n-m))} \boxminus W_{n-1}^{(N(n-m)} $}};

\node [above] at (-0.5,-0.8) {\small{$W_{{n-1}}^{(1)}$}};
\node [above] at (-0.5,-1.5) {\small{$W_{{n-1}}^{(2)}$}};

\node [above] at (-0.1,-3.3) {\small{$W_{{n-1}}^{(N(n-m))}$}};
\node [above] at (0.1,-4) {\small{$W_{{n-1}}^{(N(n-m)+1)}$}};
\node [above] at (0.1,-4.7){\small{$W_{{n-1}}^{(N(n-m)+2)}$}};

\node [above] at (-0.1,-6.1) {\small{$W_{{n-1}}^{(N(n-1))}$}};

\node [above] at (-0.4,-7.7) {\small{$\hat{W}_{{n-m}}^{(1)}$}};
\node [above] at (-0.4,-8.4) {\small{$\hat{W}_{{n-m}}^{(2)}$}};

\node [above] at (-0.05,-10.2) {\small{$\hat{W}_{n-m}^{(N(n-m))}$}};


\draw [] (-3,0.2) rectangle (2.1,-11);
\node [] at (-0.45,-11.4) {$W_n$};

\end{tikzpicture}
\end{center}
\caption
{Transition probabilities of $W_n^{(i)}$ channels after combining and splitting $W_{n-1}$ and $\hat{W}_{n-m}$. 
}
\label{fig:S_polarization}
\end{figure}
The above proposition is illustrated in Fig.~\ref{fig:S_polarization}. In order to provide a proof for the above proposition and explain the underlying idea behind the bit-reversed order we make the following analysis. Investigating Fig.~\ref{fig:S_polarization}, we see that the overall effect of XOR operations, after channel splitting, is to provide diversity paths for the $N(n-m)$ inputs of $W_{n-1}$ in the sense that 
for $ j \in \mathbb{N}_{n-m}$ we have $W_n^{(j)} = \hat{W}_{n-m}^{(j)} \boxplus  W_{n-1}^{(j)}$. Therefore the input of $W_n^{(j)}$ is transmitted through both $\hat{W}_{n-m}^{(j)}$ and $\hat{W}_{n-m}^{(j)}$. Notice that in order to provide this diversity, the inputs of $W_n^{(j+{N_{n-1}})}$ must be decoded, by the genie-aided decoder, before the inputs of $W_n^{(j)}$ indicating $\pi_n(j)> \pi_n(j+N(n-1))$ must hold. Thanks to the bit-reversed order, as explained in Definition.~\ref{defn:channel_splitting_order}, this requirement can be easily accomplished. To see this consider the state vectors $\textbf{s}_{n-1}^{(j)}$ of $W_{n-1}^{(j)}$ to which one appends $+$ and $-$ in order to construct $\textbf{s}_{n}^{(j)}$ and $\textbf{s}_{n}^{(j+N(n-1))}$, respectively. After this operation, the mapping between $\textbf{s}_n^{(i)}$ and $\textbf{b}_n^{(i)}$, as given by \eqref{eqn:binary_map}, indicates that $\textbf{b}_{n}^{(j)}=(\textbf{b}_{n-1}^{(j)},1)$ and $\textbf{b}_{n}^{(j+N(n-1))}=(\textbf{b}_{n-1}^{(j)},0)$ holds. Therefore 
\begin{align*}
(\textbf{b}_{n}^{(j)})_2 > (\textbf{b}_{n}^{(j+N(n-1))})_2, \quad n=1,2,\dots
\end{align*}
and by Definition~\ref{defn:channel_splitting_order}, $\pi_n(j) > \pi_n(j+N(n-1))$ holds for all $n\geq1$. On the other hand, in order to decode $W_n^{(j+{N_{n-1}})}$ correctly, the inputs of $W_{n-1}^{(j)}$ and $\hat{W}_{n-m}^{(j)}$ must be decoded correctly indicating we must have $W_n^{(j+N(n-1))} = \hat{W}_{n-m}^{(j)} \boxminus  W_{n-1}^{(j)}$. The above analysis, by induction through combining and splitting stages
$1,2,\ldots,n$ proves \eqref{eqn:polarized_channels}. In order to prove \eqref{eqn:non_polarized_channels}, we inspect that for $j \in  {\mathbb N}_{n-1} \setminus {\mathbb N}_{n-m}$ the channel  $W_n^{(j)}$ is as good as $ W_{n-1}^{(j)}$ in the sense that the genie-aided decoder can always decode $W_{n-1}^{(j)}$ instead of $ W_{n}^{(j)}$. Inspecting Fig.~\ref{fig:S_polarization} we notice that the binary-input of $ W_{n}^{(j)}$  is not transmitted through the inputs of $\hat{W}_{n-m}$. Therefore, the combining of $\hat{W}_{n-m}$ with $W_{n-1}$ does not provide any new information regarding the input of $W_{n}^{(j)}$. This, in turn, indicates that $W_{n}^{(j)}$ is the same as $W_{n-1}^{(j)}$ except for a scaling factor $\gamma(n)$, as in \eqref{eqn:non_polarized_channels}.

\subsection{Effects of Channel Combining and Splitting on the Symmetric Capacity}

Let us define $I_n^{(i)}=I(W_n^{(i)})$ and analyze the implications of Proposition~\ref{prop:S_split_channels}. Equation \eqref{eqn:polarized_channels} states that the channel pairs, $\hat{W}_{n-m}^{(j)}$ and $W_{n-1}^{(j)}$, $ j \in \mathbb{N}_{n-m}$, undergo a polarization transform, $\boxminus$ and $\boxplus$, from which two new channels, $W_n^{(j)}$ and $W_n^{(j+{N_{n-1}})}$, emerge. In the light of \eqref{eqn:I_plus} we have 
\begin{align}
I_n^{(j)} \geq \max \{ I_{n-1}^{(j)}, I_{n-m}^{(j)} \}, \quad j \in \mathbb{N}_{n-m}.
\end{align}
Therefore, the injection of $\hat{W}_{n-m}^{(j)}$ allows $W_{n}^{(j)}$ to be superior channel compared to $\hat{W}_{n-m}^{(j)}$ and $W_{n-1}^{(j)}$. This comes with the expense that now $W_n^{(j+N(n-1)}$ is an inferior channel compared to $\hat{W}_{n-m}^{(j)}$ and $W_{n-1}^{(j)}$  because, from  \eqref{eqn:I_minus}, one has 
\begin{align}
I_n^{(j+N(n-1))} \leq \min \{ I_{n-1}^{(j)}, I_{n-m}^{(j)} \}, \quad j \in \mathbb{N}_{n-m}.
\end{align}
Although $I_n^{(j)}$ and $I_n^{(j+N(n-1))}$ move away from $I_{n-1}^{(j)}$ and $I_{n-m}^{(j)}$,
the transformations preserve the symmetric capacity because, as  indicated by \eqref{eqn:I_sum}, we have 
\begin{align}
I_n^{(j)}+I_n^{(j+N{n-1})}=I_{n-1}^{(j)}+I_{n-m}^{(j)}, \quad j \in \mathbb{N}_{n-m}.
\end{align}
The remaining channels  $W_{n}^{(j)}$,  $ j \in  {\mathbb N}_{n-1} \setminus {\mathbb N}_{n-m}$, in Equation~\eqref{eqn:non_polarized_channels}, 
do not see any polarization transforms as their transition probabilities are scaled by $\Pr(y_{N(n-1)+1},\ldots,y_{N}) $ with respect to $W_{n-1}^{(j)}$. This scaling, in turn, results in 
\begin{align}
I_n^{(j)}=I_{n-1}^{(j)}, \quad j \in  {\mathbb N}_{n-1} \setminus {\mathbb N}_{n-m}.
\end{align}

All in all, the combining and splitting of $W_{n-1}$ and $W_{n-m}$ preserves the sum symmetric capacity as 
\begin{align}
\sum_{i \in \mathbb{N}_n}  I_n^{(i)} &= \sum_{j \in \mathbb{N}_{n-1} }  I_{n-1}^{(j)} +  \sum_{ k \in \mathbb{N}_{n-m} }  I_{n-m}^{(k)},  \label{eqn:Sum_I} 
\end{align}

\subsection{Decoding}
We will take successive cancellation decoding (SCD) of \cite{Arikan} as the default decoding method for $\{ {\mathscr C}_{n}^{(m)}\} $. The genie- aided decoder that we have explained in Section~\ref{sect:3}.B and the definition of $W_n^{(i)}$ as given by \eqref{eqn:split_channel_defn} already provide us a guideline for SCD. The only difference is, during the  calculation of  \eqref{eqn:split_channel_defn}, SCD uses its own estimates
for the vector $\textbf{u}_{n,b}^{(i)}$, which we denote as $\hat{ \textbf{u}}_{n,b}^{(i)}$.

Likelihood ratios (LRs) should be preferred in SCD so that one can eliminate the $P(y_{N_{n-1}+1}, y_{N_{n-1}+1},\ldots,y_{N_{n}})$ term in \eqref{eqn:non_polarized_channels}. The LR for the channel $W_n^{(i)}$ is defined as
\begin{align*}
L_n^{(i)} \defn \frac{ \sum_{\textbf{u}_{n,a}^{(i)} } \Pr\left( \textbf{y}_{N},\textbf{u}_{n,a}^{(i)} ,\hat{ \textbf{u}}_{n,b}^{(i)} |0 \right)}{ \sum_{\textbf{u}_{n,a}^{(i)} } \Pr\left( \textbf{y}_{N},\textbf{u}_{n,a}^{(i)} ,\hat{ \textbf{u}}_{n,b}^{(i)} |1 \right)}.
\end{align*}
By using the LR relations given in \cite{Arikan} for $\boxplus$ and $\boxminus$ transformations and from Proposition~\ref{prop:S_split_channels} we obtain
\begin{align}
\begin{split} \label{eqn:LR}
L_n^{(j)} &=  L_{n-1}^{(j)}  (L_{n-m}^{(j)})^{1-2\hat{u}_{n}^{(j+N_{n-1})}}, \\
L_n^{(j+N_{n-1})}  &=  \frac{ L_{n-1}^{(j)}  L_{n-m}^{(j)} +\! \!1 \!}
{L_{n-1}^{(j)}+L_{n-m}^{(j)}}, 
\end{split}
\!\!\!\!\!\!\!\!\!\!\!\!\!\!\!\!\!\!\!\!&j \in \mathbb{N}_{n-m}, \\ 
L_n^{(j)}&= L_{n-1}^{(j)}, &j \in \mathbb{N}_{n-1} \setminus \mathbb{N}_{n-1}.  \label{eqn:LR_sparse}
\end{align}
Therefore, while decoding $W_n^{(i)}$ one only needs to calculate $2N(n-m)$ LRs as given by \eqref{eqn:LR} while the remaining 
$N-N(n-m)$ LRs for \eqref{eqn:LR_sparse} are the same as the previous level. This fact can be exploited to avoid unnecessary decoding complexity in hardware implementation.

\subsection{Code-Length}
\! Recall that the code-length $N=N(n,m)$ obeys the recursion in \eqref{eqn:code_length_recursion} with initial conditions of \eqref{eqn:code_length_init}. It is easy to show that $N$ can be calculated as
\begin{align}
N =  \sum_{i=1}^m c_i (\rho_i)^n , \label{eqn:code_length_root_sum}
\end{align}
where each $\rho_i$, $i=1,2,\ldots,m$, is a root of the $m$th order polynomial equation 
\begin{align}
F(m, \rho)=\rho^m - \rho^{m-1}-1,\label{eqn:roots_of_N}
\end{align}
and constants, $c_i$, are calculated by using the initial conditions in \eqref{eqn:code_length_init} together with \eqref{eqn:code_length_root_sum}.
\begin{proposition}\label{prop:Dominating_root}
For $m \geq 1$, let $\phi \in (1,2]$ be a real root of $F(m, \rho)$. 
\begin{enumerate}  
\item  $\phi$ is unique, i.e., there is only one real root in $\in (1,2]$.
\item If $\rho_i \neq \phi$ we have $ \sqrt{\rho_i \rho_i^*}/\phi<1$  indicating $\phi$ is the the largest magnitude root of 
$F(m, \rho)$.
\item $\phi$ is decreasing in increasing $m$.
\end{enumerate}
\end{proposition}
Part \textit{ii} of the above proposition indicates that, as $n$ gets large, the summation in \eqref{eqn:code_length_root_sum} will be dominated by $\phi^n$ term therefore the code-length will scale as $N=\kappa\phi^n=O(\phi^n)$ where $\kappa>0$ is the constant scaler of $\phi^n$ in \eqref{eqn:code_length_root_sum}. Part \textit{iii} of Proposition~\ref{prop:Dominating_root} implies that  as $m$ increases the code-length increases less rapidly in $n$ which we have mentioned in the beginning of the paper.

\subsection{Code Construction}

The following proposition is a generalization of  \cite[Prop. 5]{Arikan} and it's proof is omitted.
\begin{proposition}\label{prop:S_split_channels_are_BEC}
If $W$ is a BEC, then $W_n^{(i)}$ channels obeying the transition probabilities as given by Proposition~\ref{prop:S_split_channels} are also BECs.
\end{proposition}

In order to use $\{ {\mathscr C}_{n}^{(m)}\}$ one has to fix a code parameter vector $(W,N,K,{\cal A})$, where $W$ is the underlying B-DMC, $N$ is the code-length, $K$ is the dimensionality of the code, and ${\cal A} \subseteq {\mathbb N}_n$ is the set of information carrying symbols.  We have ${|{\cal A} |}=K$ and $K/N=R$, where $R \in [0,1]$ is the rate of the code. 

Let $P_{e,n}^{(i)}$, $ i \in {\mathbb N}_n$, denote the bit-error probability of $W_n^{(i)}$ with SCD. Code construction problem is choosing the set ${{\cal A} }$ so that $\sum_{i \in {\cal A}} P_{e,n}^{(i)}$ is minimum. This problem can be analytically solved only when $W$ is a BEC \cite{Arikan} since for this case the $W_n^{(i)}$ channels are also BECs  (Proposition \ref{prop:S_split_channels_are_BEC}) and
the Bhattacaryya parameters of $W_n^{(i)}$, which we denote as $Z_n^{(i)}$, obey $P_{e,n}^{(i)}=Z_n^{(i)}$. In this case, in the light of \eqref{eqn:Z_plus}-\eqref{eqn:Z_minus} and Proposition~\ref{prop:S_split_channels}, $Z_n^{(i)}$ terms can be recursively calculated as
\begin{align*}
\begin{split}
Z_n^{(j)}  &= Z_{n-1}^{(j)} Z_{n-m}^{(j)}, \\
Z_n^{(j+N_{n-1})}  &= Z_{n-1}^{(j)}+ Z_{n-m}^{(j)} -Z_{n-1}^{(j)}+Z_{n-m}^{(j)},
\end{split} \quad  j \in \mathbb{N}_{n-m}, \\
Z_n^{(j)} &= Z_{n-1}^{(j)} \quad \quad \quad \quad \quad   j \in \mathbb{N}_{n-1} \setminus \mathbb{N}_{n-1}. 
\end{align*}
The case when $W$ is not a BEC is a well-studied problem, where one approximates a suitable reliability measure for $W_n^{(i)}$ channels and uses this measure to choose the set ${\cal A}$. We refer the reader to \cite{Tal_construction} for an overview.

\section{Channel Polarization}\label{sect:4}
Channel polarization should be investigated by observing the evolution of the set $ \{ W_n^{(i)}: i \in \mathbb{N}_n \}$ as $n$ increases. To track this evolution we use the state vectors $\textbf{s}_n^{(i)} \in {\cal S}_n$ assigned to $ W_n^{(i)}$ because each $W_n^{(i)}$ is uniquely addressable by its $\textbf{s}_n^{(i)}$.

\subsection{Probabilistic Model for Channel Evolution}
We define a random process $\{ S_n\}$ and a random vector $\textbf{S}_n =(S_1,S_2,\ldots,S_n)$ obtained from the process $\{ S_n\}$ where the state vectors, $\textbf{s}_n=(s_1,s_2,\ldots,s_n)$, $\textbf{s}_n \in {\cal S}_n$, of Section~\ref{sect:2}, are the realizations of $\textbf{S}_n$. The process $\{ S_n\}$ can be regarded as a tree process where $\textbf{s}_n$ form the branches of the tree where we illustrate it in Fig.~\ref{fig:tree_process} for the case $m=2$. Since $|{\cal S}_n|=N=N(n)$, there are $N(n)$ different branches at tree level $n$. The process $\{ S_{n} \}$ starts with the initial conditions $S_1 \in \{+,-\}$. At tree level $n$, $N(n)$ new branches emerge from $N(n-1)$ branches of level $n-1$. We assume that each branch is observed with identical probability
\begin{align}\label{eqn:state_prob}
\Pr (\textbf{S}_n=\textbf{s}_n)=\frac{1}{N(n)}.
\end{align}
This, in turn, implies that each valid state transition of Fig.~\ref{fig:S_state_transition}, between $s_{n-1}$ and $s_{n}$, has probability $N(n-1)/{N(n)}$. Investigating this figure, consider the case $m=1$, which coincides with Ar\i kan's setup in \cite{Arikan}, where there are two possible states as $S_n \in \{+,-\}$ and 
$|{\cal S}_n|=N(n)=2^n$. Since transitions between $S_{n-1}$
and $S_n$ are valid if $S_{n} \in \{+,-\}$ and $S_{n-1}\in \{+,-\}$,  each possible transition has probability 
$N(n-1)/N(n)=1/2$. Consequently, the process $\{ S_n\}$ is composed of independent realizations of Bernoulli$(1/2)$ random variables as $\Pr({S_n=+})=\Pr({S_n=-})=1/2$. On the other hand, when $m>1$, there exists a memory in the state transition model as depicted in Fig.~\ref{fig:S_state_transition}. Therefore, the process 
$\{ S_n \}$ can be modeled as a Markov process with order $m-1$ in the sense that
\begin{align*}
&\Pr (S_n|\textbf{S}_{n-1}) =  \Pr (S_n|S_{n-1},S_{n-2},\ldots,S_{n-(m-1)}).
\end{align*}
Throughout the paper we find it easier to  work with the random vector $\textbf{S}_n$ keeping in mind the Markovian property of the process $\{ S_n \}$.

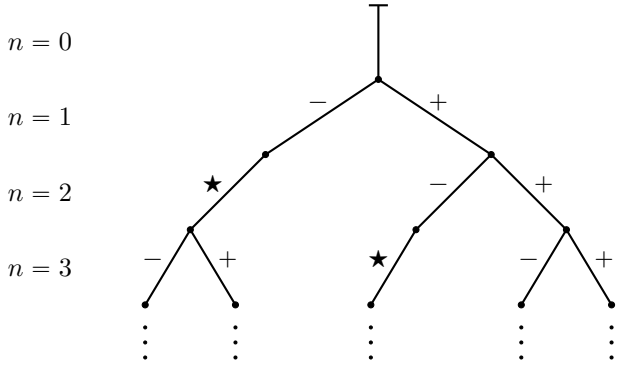
\begin{figure}[tb]
\begin{center}

\begin{tikzpicture}

\draw [thick] [|-] (0,1) -- (0,0);
\node [] at (-4.5,0.5) {$n=0$};
\node [] at (-4.5,-.5) {$n=1$};
\node [] at (-4.5,-1.5) {$n=2$};
\node [] at (-4.5,-2.5) {$n=3$};

\draw [fill] (0,0) circle [radius=0.04];

\draw [thick] [-] (0,0) -- (1.5,-1);
\node [] at (0.8,-0.3) {$+$};
\draw [fill] (1.5,-1) circle [radius=0.04];

\draw [thick] [-] (0,0) -- (-1.5,-1);
\node [] at (-0.8,-0.3) {$-$};
\draw [fill] (-1.5,-1) circle [radius=0.04];
\draw [thick] [-] (1.5,-1) -- (0.5,-2);
\node [] at (0.8,-1.4) {$-$};
\draw [fill] (0.5,-2) circle [radius=0.04];

\draw [thick] [-] (1.5,-1) -- (2.5,-2);
\node [] at (2.2,-1.4) {$+$};
\draw [fill] (2.5,-2) circle [radius=0.04];

\draw [thick] [-] (-1.5,-1) -- (-2.5,-2);
\node [] at (-2.2,-1.4) {$\bigstar$};
\draw [fill] (-2.5,-2) circle [radius=0.04];
\draw [thick] [-] (0.5,-2) -- (-0.1,-3);
\node [] at (0,-2.4) {$\bigstar$};
\draw [fill] (-0.1,-3) circle [radius=0.04];

\draw [thick] [-] (2.5,-2) -- (3.1,-3);
\node [] at (3,-2.4) {$+$};
\draw [fill] (3.1,-3) circle [radius=0.04];

\draw [thick] [-] (2.5,-2) -- (1.9,-3);
\node [] at (2,-2.4) {$-$};
\draw [fill] (1.9,-3) circle [radius=0.04];

\draw [thick] [-] (-2.5,-2) -- (-3.1,-3);
\node [] at (-3,-2.4) {$-$};
\draw [fill] (-3.1,-3) circle [radius=0.04];

\draw [thick] [-] (-2.5,-2) -- (-1.9,-3);
\node [] at (-2,-2.4) {$+$};
\draw [fill] (-1.9,-3) circle [radius=0.04];

\draw [fill] (-0.1,-3.3) circle [radius=0.02];
\draw [fill] (-0.1,-3.5) circle [radius=0.02];
\draw [fill] (-0.1,-3.7) circle [radius=0.02];

\draw [fill] (3.1,-3.3) circle [radius=0.02];
\draw [fill] (3.1,-3.5) circle [radius=0.02];
\draw [fill] (3.1,-3.7) circle [radius=0.02];

\draw [fill] (-1.9,-3.3) circle [radius=0.02];
\draw [fill] (-1.9,-3.5) circle [radius=0.02];
\draw [fill] (-1.9,-3.7) circle [radius=0.02];

\draw [fill] (1.9,-3.3) circle [radius=0.02];
\draw [fill] (1.9,-3.5) circle [radius=0.02];
\draw [fill] (1.9,-3.7) circle [radius=0.02];

\draw [fill] (-3.1,-3.3) circle [radius=0.02];
\draw [fill] (-3.1,-3.5) circle [radius=0.02];
\draw [fill] (-3.1,-3.7) circle [radius=0.02];

\end{tikzpicture}
\end{center}
\caption{Illustration of the evolution of $\{S_n\}$ as a tree for the case $m=2$, where each branch is a state vector $\textbf{s}_n \in {\cal S}_n$.}
\label{fig:tree_process}
\end{figure}

We define a random channel process $\{K_n\}$, driven by $\{ S_{n} \}$, as $K_n=W_{S_1,S_2,\ldots,S_n}$. The realizations of $K_n$ are $k_n= W_{s_1,s_2,\ldots,s_n}$ and they correspond to the binary-input channels, $W_n^{(i)}$, with state vectors $\textbf{s}_n=(s_1,s_2,\ldots,s_n) \in {\cal S}_n$.

In order to obtain a characterization for the process $\{K_n\}$ we fix $(s_1,s_2,\dots,s_{n-1})$ to be the state vector associated with $W_{n-1}^{(j)}$, $j \in \mathbb{N}_{n-m}$ and let  $k_{n-1}=W_{n-1}^{(j)}$. In the light of Proposition~\ref{prop:same_state_history}, we know that the state vector of $\hat{W}_{n-m}^{(j)}$ is $(s_1,s_2,\dots,s_{n-m})$ indicating $k_{n-m}=\hat{W}_{n-m}^{(j)}$. Investigating the operation of $\varphi_n: {\cal S}_{n-1} \rightarrow {\cal S}_n$ in Fig.~\ref{fig:S_state_labeling}, we observe that the state vectors of $W_{n}^{(j)}$ and $W_{n}^{(j+N_{n-1})}$ are $(s_1,s_2,\dots,s_{n-1},+)$ and $(s_1,s_2,\dots,s_{n-1},-)$, respectively. From Proposition~\ref{prop:S_split_channels} we notice that $W_{n}^{(j)}=\hat{W}_{n-m}^{(j)} \boxplus W_{n-1}^{(j)}$ and $W_{n}^{(j+N(n-1))}=\hat{W}_{n-m}^{(j)} \boxminus W_{n-1}^{(j)}$ holds. These observations, in turn, indicate $k_n=k_{n-1} \boxplus k_{n-m}$ holds when $s_n=+$, and $k_n=k_{n-1} \boxminus k_{n-m}$ holds when $s_n=-$. Next, we fix $(s_1,s_2,\dots,s_{n-1})$ to be the state vector associated with $W_{n-1}^{(j)}$, $j \in \mathbb{N}_{n-1} \setminus \mathbb{N}_{n-m}$ and hence $k_{n-1}=W_{n-1}^{(j)}$. From the operation of  $\varphi_n: {\cal S}_{n-1} \rightarrow {\cal S}_n$ we know that the state vector of $W_{n}^{(j)}$
is $(s_1,s_,\ldots,s_{n-1},\bigstar)$ and Proposition~\ref{prop:S_split_channels} tells us $W_{n}^{(j)}=\gamma(n)W_{n-1}^{(j)}$. Combining these
facts tells us $k_n=\gamma(n)k_{n-1}$ holds if $s_n=\bigstar$. The above analysis relates $k_n$ to $k_{n-1}$ and $k_{n-m}$ for all
$s_n \in \{+,-,\bigstar\}$, which we formally present with the below recursion.
\begin{align}\label{eqn:channel_evolution_model}
K_n = 
\begin{cases}
K_{n-m} \boxplus  K_{n-1} \quad \text{if}  \quad  S_n=+, \\
K_{n-m} \boxminus K_{n-1} \quad \text{if} \quad  S_n=- ,  \\
\gamma(n)K_{n-1}  \quad \text{otherwise},
\end{cases}
\end{align}
where $K_n=W$ for $n<1$. 

\subsection{Polarization:}
We define the processes $\{I_n: n \geq 1\}$ and $\{J_n : n \geq 1\}$ where  $I_n = I(K_n) \in [0,1]$ and $J_n=J(K_n) \in [0,1]$. In \cite{Arikan} Ar\i kan shows that $I_n$ converges to a random variable $I_{\infty}$ as $\Pr(I_{\infty}=1)=I(W) $  and $\Pr(I_{\infty}=0)=1-I(W)$. This result indicates that
the synthesized binary-input channels, $W_n^{(i)}$, either become error-free or useless. We will show that the same holds for polar codes with higher-order memory as well. This result is presented with the following theorem.
\begin{theorem}\label{thm:achieving_capacity}
For any fixed $m\geq 1$ and for some $\delta \in (0,1)$ as $n$ tends to infinity, the probability of $I_n\in (1-\delta,1]$ goes to $I(W)$ and the probability of having $I_n \in [0,\delta)$  goes to $1-I(W).$
\end{theorem}

\begin{proof}
We investigate the polarization of $\{J_n \}$ towards $0$ and $1$ as it will imply the polarization of $\{I_n \}$ as well.
We write $E[J_n]=\sum_{ \textbf{s}_n} \Pr(\textbf{S}_n =\textbf{s}_n)J_n=\frac{1}{N(n)}\sum_{ \textbf{s}_n } J_n$ to denote the expected value of $J_n$ and $\{ E[J_n] :n \geq 1\}$ to denote the deterministic sequences obtained from $E[J_n]$. The following lemma will be crucial for the proof
\begin{lemma}\label{lemma:sum_cut_off_eq}
\begin{align}
\!\! E[J_n]  \geq \mu E[J_{n-1}] + \left(1-\mu \right) E[J_{n-m}] ,  \label{eqn:j_sum_1}
\end{align}
where $\mu=N(n-1)/N(n)$ and the above equality is achieved only if $J_{n-1} \in \{0,1\}$ or $J_{n-m}  \in \{0,1\}$ holds for all $S_n \in \{+,-\}$
\end{lemma}

We apply a decimation operation on the sequence $\{ E[J_n]  \}$ and obtain a subsequence $\{ E[ \hat{J}_k] : 
k=1,2,\dots, \lfloor n/m \rfloor \}$, where the decimation operation is performed as 
\begin{align}\label{eqn:decimated_process}
\! \! \! \!   E[ \hat{J}_k]  = \min_{i \in \{0,1,\ldots,m-1\}} \left \{ E[J_{km-i} ] \right \}. 
\end{align}
The elements of $\{ E[ \hat{J}_k] \}$ are obtained by choosing the minimum of $m$ consecutive and non-overlapping elements of
$\{ E[J_n] \}$.
\begin{lemma}\label{prop:min_symmetric_cut_off}
The sequence $\{ E[ \hat{J}_k] \}$ is monotonically increasing in the sense that
\begin{align*}
E[ \hat{J}_k]  \geq E[ \hat{J}_{k-1}] . 
\end{align*} 
\end{lemma}
We know that $E[ \hat{J}_k]$ is bounded in $[0,1]$ and since $\{ E[ \hat{J}_k] \}$ is monotonically increasing, from the monotone convergence theorem \cite[p. 21.]{Bartle} we conclude that there exists a unique limit for  $\{ E[ \hat{J}_k] \}$ in the sense that 
\begin{align}
\lim_{k \rightarrow \infty} E[ \hat{J}_k] = \sup \{ E[ \hat{J}_k] \}. \label{eqn:cut_off_limit}
\end{align} 
Next, we let $n=km-i$ in Lemma \ref{lemma:sum_cut_off_eq} to obtain
\begin{align}
 E[J_{km-i}]  &\geq \mu E[J_{km-(i+1)}] + (1-\mu) E[J_{(k-1)m-i}].  \label{eqn:cut_off_weighted}
\end{align}
We fix $i$ such that $E[J_{km-i}]= E[\hat{J}_{k}]$ is satisfied. For any choice of $i$ observe that $E[J_{(k-1)m-i)}] \geq E[\hat{J}_{k-1}]$ and 
$ E[J_{km-(i+1)}] \geq \min\{   E[\hat{J}_{k}], E[\hat{J}_{k-1}] \} \geq E[\hat{J}_{k-1}] $ hold. Using these results in  \eqref{eqn:cut_off_weighted} 
gives
\begin{align}
E[\hat{J}_{k}] \geq \mu E[\hat{J}_{k-1}] + (1-\mu)E[\hat{J}_{k-1}] \geq E[\hat{J}_{k-1}] 
\end{align}
Therefore, the monotonic increase in $E[\hat{J}_{k}]$ will continue until the inequality in Lemma \ref{lemma:sum_cut_off_eq} is achieved with equality. This fact, together with the convergence of $E[ \hat{J}_k]$, indicates that conditioned on the event $\{  S_n : S_n \in \{+,-\} \}$ either $\lim_{n \rightarrow \infty} J_{n-1} \in \{ 0,1 \}$ or $\lim_{n \rightarrow \infty} J_{n-m} \in \{ 0,1 \}$  holds, indicating 
\begin{align}
\lim_{n \rightarrow \infty} J_{n} \in \{ 0,1 \}, \quad        S_n \in \{+,-\}.
\end{align}
Investigating the operation of $\varphi_n: {\cal S}_{n-1} \rightarrow {\cal S}_n$ in  Fig.\ref{fig:S_state_labeling} we see that
\begin{align}
\Pr \left( S_{n} \in \{+,-\} \right) =\frac{2N(n-m)}{N(n)}\geq 0,
\end{align}
which implies that the event $ \{ S_n: S_{n-1} \in \{+,-\} \}$ occurs infinitely many times as $n \rightarrow \infty$ and $ \sum_{n \rightarrow \infty} \Pr \left( S_{n-1} \in \{+,-\} \right) $ diverges. Consequently, and by using the first Borel Contelli lemma \cite[p. 36]{Billingsley} we conclude that
\begin{align*}
\lim_{n \rightarrow \infty} \Pr (  J_n \in \{ 0,1 \} ) =1.
\end{align*}
One to one correspondence between $J_n$ and $I_n$ implies 
\begin{align*}
\lim_{n \rightarrow \infty} \Pr (  I_n \in \{ 0,1 \} )=1,
\end{align*}
and having $E [ I_n] = I(W)$ results in
\begin{align*}
\lim_{n \rightarrow \infty} \Pr (  I_n =1 )=I(W),
\end{align*}
and
\begin{align*}
\lim_{n \rightarrow \infty} \Pr (  I_n =0 )=1-I(W).
\end{align*}
which completes the proof.
\end{proof}

\subsection{A Typicality Result}
In this section we use the Method of Types to investigate the state vectors, $\textbf{s}_n$, obtained from the realizations of the process $\{S_n\}$. 
We let $s \in \{+,-,\bigstar \}$ and write $P_{\textbf{s}_n}^{(s)}$, $P_{\textbf{s}_n}^{(s)} \in [0,1]$, to denote the type (frequency) of $s$ in $\textbf{s}_n$ as
\begin{align*}
P_{\textbf{s}_n}^{(s)} =  \#(\textbf{s}_n|s)/n, 
\end{align*}
where $\#(\textbf{s}_n|s)$ denotes the number times the symbol $s$ occurs in $\textbf{s}_n$. Investigating the state transition diagram of
Fig.~\ref{fig:S_state_transition} we inspect that, as $n$ gets large, $P_{\textbf{s}_n}^{(\bigstar)}=(m-1)P_{\textbf{s}_n}^{(-)}$ holds because
each $-$ state in $\textbf{s}_n$ is followed by $m-1$ occurrences of state $\bigstar$. As the remaining states in $\textbf{s}_n$ will be $+$, we must have $P_{\textbf{s}_n}^{(+)}=1-mP_{\textbf{s}_n}^{(-)}$ indicating $P_{\textbf{s}_n}^{(+)} \in [0,1]$,  $P_{\textbf{s}_n}^{(-)} \in [0,\frac{1}{m}]$, and $P_{\textbf{s}_n}^{(\bigstar)} \in [0,\frac{m-1}{m}]$. As it tuns out, depending on $P_{\textbf{s}^n}^{(s)}$, not all realizations of $\{S_n\}$ are observed with the same probability. This is explained with the following theorem.
\begin{theorem}\label{thm:typical_freq}
As $n$ gets large, except for a vanishing fraction of $\textbf{s}_n \in {\cal S}_n$, and for some $\epsilon \in (0,1)$ we have
\begin{align*}
 |P_{\textbf{s}^n}^{(\minus)}-p^-|    & \leq \epsilon,  \\
 |P_{\textbf{s}^n}^{(\plus)}-p^+|  &  \leq \epsilon, \\
 |P_{\textbf{s}^n}^{(\bigstar)}-p^\bigstar|   & \leq \epsilon,
\end{align*}
where $p^- = \frac{\phi-1}{1+m(\phi-1)}$, $p^\bigstar=(m-1)p^-$  and $p^+=1-mp^-$.
\end{theorem}
Therefore we can consider $p^+$, $p^-$ and $p^\bigstar$ as the frequencies of states $+$, $-$, and $\bigstar$, in $\textbf{s}_n$, respectively, that one typically observes as $n$ gets large.  

\vspace{5pt}
\textit{\textbf{Proof of Theorem \ref{thm:typical_freq} :}} 
The proof is based on the Method of Types \cite{Cover}. We let $q \in [0,1/m]$ and define
\begin{align}
{\cal T}_n^{(q)} =\{ \textbf{s}^n : P_{\textbf{s}^n}^{(\minus)}=q \} \label{eqn:typical_set_defn}.
\end{align}
${\cal T}_n^{(q)}$ is a type class and it consists of $\textbf{s}_n$ having $nq \in [0, n/m]$ occurrences of state $-$. For all $m\geq 1$, there are at most $n+1$ different such type classes. However, the number of all possible $\textbf{s}_n$, $|{\cal S}_n|$, increases exponentially in $n$ as $|{\cal S}_n|=N=O(\phi^n)$. The Method of Types ensures
the existence of a type class with exponentially many elements. Our aim is to find this type class. 
Recalling that each $\textbf{s}_n$ is observed with probability $1/N$, the probability of observing a given $\textbf{s}_n$ in ${\cal T}_n^{(q)}$ is 
\begin{align*}
\Pr \left( \textbf{s}_n \in {\cal T}^n_q  \right)  = \frac{| {\cal T}^n_q |}{N}.
\end{align*}
\begin{lemma}\label{lemma:class_size_upper_bound}
\begin{align}\label{eqn:size_type}
  |{\cal T}^n_q| < 2^{  n\left( G(m,q)+o(1) \right)   }.
\end{align}
where 
\begin{align*}
G(m,q) =  (1-(m-1)q)) H\left(\frac{q}{1-(m-1)q}\right),
\end{align*}
and $H$ is the binary entropy function.
\end{lemma}
Investigating $G(m,q)$ we observe that it is a concave function of $q \in [0,1/m]$. We establish a similarity between $\frac{ \partial  G(m,q)}{\partial q}$ and $F(m,\rho)$ in \eqref{eqn:roots_of_N}. The following proposition is a direct consequence of this result.
\begin{lemma}\label{lemma:root_and_bound_duality}
The function $G(m,q)$ attains its maximum when  $q=p^-$ and its maximum value is 
\begin{align*}
G(m,p^-) =\! \log \phi.
\end{align*}
\end{lemma}
Consequently, for every ${\cal T}_n^{(q)}$ with $|q-p^-| > 0$ there exists a $D(q, p^-) > 0$ such 
that 
\begin{align*}
D(q,p^-) &  \defn G(m,p^-) - G(m,q), \\
                & = \log \phi - G(m,q). 		
\end{align*}
Using the above fact in $\eqref{eqn:size_type}$ results in
\begin{align*}
   |{\cal T}_n^{(q)}|  \leq \phi^n 2^{n \left(-D(q,p^-) +o(1) \right)}. 
\end{align*}
From the above result and the fact that $N=O(\phi^n)$ we obtain
\begin{align}
 \Pr (\textbf{s}_n \in {\cal T}_n^{(q)})   \leq 2^{-n\left(D(q,p^-)+o(1) \right)},  \label{eqn:pTk}
\end{align}
The above result shows that depending on $D(q,p^-)$, and in turn $q$, the probabilities of some type classes decay exponentially in $n$. The following proposition results from this fact.
\begin{proposition}\label{prop:typical_set_S}
As $n$ tends to infinity $D(q,p^-)$ converges to $0$ with probability $1$.
\end{proposition}
The above proposition implies the convergence of $q$ to $p^-$ as well, because $D(q,p^-)$ is $0$ only if $q=p^-$. Therefore among all $T_n^{(q)}$, one observes the ones with $|q-p^-| \leq \epsilon$ with probability $1$.

\subsection{Rate of Polarization}

We define the Bhattacharyya process $\{Z_n\}$ where $Z_n=Z(K_n)$ is the  Bhattacharyya parameter of the random channel $K_n$.
By using the channel evolution model in \eqref{eqn:channel_evolution_model}, this process can be expressed as
\begin{align}\label{eqn:Bhat_Process}
 Z_n  \begin{cases} = Z_{n-1}Z_{n-m}  & \text{if $S_{n}=+$},\\
  			\leq Z_{n-1}+Z_{n-m}-Z_{n-1}Z_{n-m}  & \text{if $S_{n}=-$},\\
		           = Z_{n-1} & \text{otherwise},
  \end{cases}
\end{align}
where $Z_n=Z(W)$ for $n<1$.

\begin{theorem}\label{thm:error_prob_bound}
For any $\epsilon \in (0,1)$ there exists an $n$ such that for $\beta< p^+$ we have
\begin{align}
        \Pr \left( Z_n \leq 2^{-\phi^{n\beta}} \right) \geq I(W)-\epsilon,  \label{eqn:lower_bound}
\end{align}
\end{theorem} 
\begin{proof}
We consider another process $\{ \hat{Z}_n \}$, driven by $\{ S_n \}$, so that for $i=1,2,\ldots,n_0$,
$n_0<n$, we have $\hat{Z}_i=Z_i$ and for $i>n_0$, $\hat{Z}_i$ obeys
 \begin{align}\label{eqn:Bhat_Process_1}
 \hat{Z}_i  =\begin{cases}  \hat{Z}_{i-1}\hat{Z}_{i-m}  & \text{if $S_{n}=+$},\\
  			  \hat{Z}_{i-1}+\hat{Z}_{i-m}-\hat{Z}_{i-1}\hat{Z}_{i-m}  & \text{if $S_{n}=-$},\\
		            \hat{Z}_{i-1} & \text{otherwise}.
  \end{cases}
\end{align}
Comparing \eqref{eqn:Bhat_Process} and \eqref{eqn:Bhat_Process_1} we observe that ${Z_n}$ is stochastically dominated by $\hat{Z}_n$ in the sense that for some $f_n \in (0,1)$, $\Pr(Z_n \leq f_n) \geq \Pr(\hat{Z}_n \leq f_n)$. For the proof it will suffice to show that $\Pr(\hat{Z}_n \leq f_n) \geq I(W)-\epsilon$ holds for $f_n = 2^{-\phi^{n\beta}}$ and $\beta<p_+$.

In \cite[Lemma  1]{Afser_Delic} authors derive an upper bound on $\hat{Z}_n$, for the case $m=1$, by using the frequency of state $+$ in 
the realizations of $\{ S_{n_0+1},S_{n_0+2},\ldots,S_{n} \}$ and the fact that $Z_{n_0}$ gets arbitrarily close to $0$, with probability $I(W)$,
when $n_0$ is large enough. Following lemma is a generalization of  this approach for arbitrary $m \geq 1$.

\begin{lemma}\label{lemma:direct_bound}
For some $\zeta \in (0,1)$ and $\gamma \in (0,1)$ define the events 
\begin{align*}
C_{n_0}(\zeta) &= \{ Z_{n_0} \leq \zeta \}, \\
D_{n_0}^n(\gamma) &= \{ \#\big( (S_{n_0+1},\ldots,S_{n})|+\big) \geq \gamma(n-n_0) \}.  
\end{align*}
We have
\begin{align*}
     \hat{Z}_n \leq 2^{-\phi^{(\gamma-\epsilon)(n-n_0)}}, \quad C_{n_0}(\zeta)  \cap D_{n_0}^n(\gamma). 
\end{align*}
\end{lemma}
From the convergence of $Z_n$ to $Z_{\infty}$ with probability $\Pr(Z_{\infty}=0)=I(W)$ we know that for any $\epsilon \in (0,1)$ there exist a fixed $n_0$ such that
\begin{align*}
\Pr ( C_{n_0}(\zeta) ) \geq I(W)-\epsilon.
\end{align*}
Next, from Theorem~\ref{thm:typical_freq}, we infer that when $m \ll n-n_0$
\begin{align}
\Pr (  D_{n_0}^n(\gamma) ) \geq 1-\epsilon, \quad \gamma \geq p^+-\epsilon
\end{align}  
holds. This results from the fact that the probability of observing $+$ in $\{ S_{n_0+1},\ldots,S_{n_0} \}$ approaches to $p^+$ when
$n-n_0$  is much larger than the memory, $m$, of the process $\{S_n\}$.

Choosing $n_0=n\epsilon$ and using the above results in lemma~\ref{lemma:direct_bound} gives
\begin{align*}
\Pr \left( \hat{Z}_n \leq 2^{-\phi^{n(p^+-2\epsilon)(1-\epsilon)}} \right) &\geq (1-\epsilon)(I(W)-\epsilon) \\
&\geq I(W)-\epsilon
\end{align*}
Since $\epsilon \in (0,1)$ can be chosen arbtirarily close to $0$, the above result indicates that
\begin{align*}
\Pr \left( \hat{Z}_n \leq 2^{-\phi^{n\beta}} \right)&\geq I(W)-\epsilon  
\end{align*}
holds for $\beta<p^+$.
\end{proof}
Let us analyze the implications of  Theorem \ref{thm:error_prob_bound} on the block-decoding error probability, $P_e$, of  $\{ \mathscr{C}_n^{(m)} \}$. It states that for $I(W)-\epsilon$ fraction of $W_n^{(i)}$ the corresponding Bhattacharyya parameters
will be bounded as $Z_n^{(i)} \leq 2^{-\phi^{n \beta }}$ for $\beta<p^+$. We have 
$P_e \leq \sum_{i=1}^NZ_n^{(i)} \leq N2^{-\phi^{n\beta}}=O(2^{-\phi^{n\beta}})$. Since the code-length of $\{ \mathscr{C}_n^{(m)} \}$ scales as $N=O(\phi^n)$ we also see that $P_e =O(2^{-N^\beta})$ holds for $\beta<p^+$.

The term $p^+$  is plotted in Fig.~\ref{fig:S_exponent_bounds} as a $m$ increases from $1$ to $50$. Investigating this figure  we see that $p^+$ equals to $0.5$ when $m=1$ which coincides with the bound for the exponent of polar codes presented by Ar\i kan and Telatar in  \cite{Arikan_rate}. As $m$ increases  from $1$ to $50$, $p^+$ and thus the achievable exponent decreases. The decrease is more steep for small values of $m$ and it becomes more monotone as $m$ increases.

 In order to fully characterize the asymptotic performance of $\{ \mathscr{C}_n^{(m)} \}$ one needs to provide a converse bound on $\beta$ which may be a difficult task. We believe that for the case $m>1$, the achievable $\beta$  for $\{ \mathscr{C}_n^{(m)} \}$ may show a dependency on the rate, $R \in [0,1]$, chosen for the code; a phenomenon that does not exist when $m=1$ 
(see \cite{urbanke_scaling_1}). In order explain our conjecture, consider the process $\{ \hat{Z}_n \} $ in \eqref{eqn:Bhat_Process_1} which we use
to obtain an achievable  bound on $\beta$ as $\beta<p^+$. Our proof is based on the observation that once the realizations of $\hat{Z}_{n_0}$ are sufficiently close to $0$, which happens with probability $I(W)$, the scaling of $Z_n$ is mostly determined by the number of occurrences of state $+$ in $\{ S_{n_0+1},S_{n_0+2},\ldots,S_{n} \}$.  From Theorem~\ref{thm:typical_freq} we know that one typically observes $(n-n_0)p^+$ occurrences of $+$ in $\{ S_{n_0+1},S_{n_0+2},\ldots,S_{n} \}$, therefore the value of  $\log Z_n$ decreases $(n-n_0)p^+$ times with the same speed as the code-length, $\log \hat{Z}_n = \log \hat{Z}_{n-1}+\log \hat{Z}_{n-m}$, scaling as $\log Z_n = - \phi^{(n-n_0)p^+} =- \phi^{n(1-\epsilon)p^+}$. This result in the achievable exponent 
$\beta<p^+$. However, when $m>1$ the value of $\log \hat{Z}_n$ may also decrease with a faster rate compared to that  of the code-length. To see this, consider the case  $(S_{n-1},S_{n-2},\ldots,S_{n-(m-1)})= ( \bigstar, \bigstar,\ldots,\bigstar )\}$ and $S_n=+$,  where we have  $\hat{Z}_{n-1}=\hat{Z}_{n-2}=\ldots=\hat{Z}_{n-(m-1)}$ and $ \log \hat{Z}_n=\log \hat{Z}_{n-1}+\log \hat{Z}_{n-m}=\log \hat{Z}_{n-1}^2$. Therefore,
there may be times where $\log Z_n$ decreases with a faster rate as $\log \hat{Z}_{n}=\log Z_{n-1}^2$ instead of $\log \hat{Z}_{n}=\log  \hat{Z}_{n-1}+\log \hat{Z}_{n-m}$ and this may result in a higher achievable $\beta$. In order to quantify this we need to know not only 
the number of times state $+$ occurs in $\{S_n\}$, but also the number of times a state $+$  in $\{S_n\}$ is preceded by $\bigstar$ states. Therefore, we need to refine Theorem~\ref{thm:typical_freq} in terms of the number of transitions between states $+$, $-$ and $\bigstar$, as well. This might be  a difficult but important problem whose solution will provide a full characterization of the asymptotic polarization performance of $\{ \mathscr{C}_n^{(m)} \}$ and we leave it as a future work.

\begin{figure}[tb]
\includegraphics[scale=0.55]{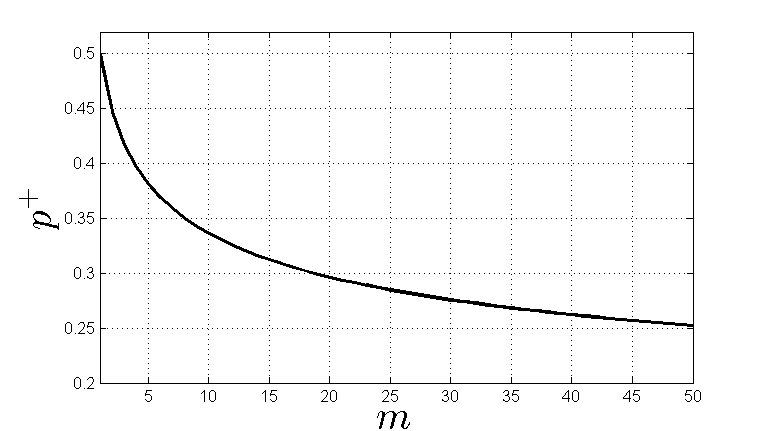}
 \caption{Achievable exponent, $\beta<p^+$, as scaled with $m$.}
\label{fig:S_exponent_bounds}
\end{figure}

\section{Complexity and Sparsity}\label{sect:5}

\subsection{Encoding and Decoding Complexity}
We consider a single core processor with random access memory and investigate the time complexity of encoding and decoding of $\{ \mathscr{C}_n^{(m)} \}$.  Let $\chi_{n}^{E}$ denote the complexity for encoding the information vector $\textbf{u}_{N}$ to encoded bits $\textbf{x}_{N}$. We take complexity of each XOR operation as $1$ unit. By inspection of Fig~\ref{fig:S_channel_combining}, we have
\begin{align}
\chi_{n}^{E} = \chi_{n-1}^{E} + \chi_{n-m}^{E}  + N_{n-m} \quad n,m \geq 1, \label{eqn:enc_complexity_recursion}
\end{align}
where $\chi_{1}^{E}=1$ and $\chi_{0}^{E}=\chi_{-1}^{E}=\ldots =\chi_{1-m}^{E}=0$.

Similarly, let $\chi_{n}^{D}$ denote the complexity for decoding the inputs of $W_n^{(i)}$ channels, where SCD is the decoding method. 
We take the complexity of computing the LR. relations in \eqref{eqn:LR} as $1$ unit. We observe that one does not make any operations to calculate the LR in \eqref{eqn:LR_sparse}. By inspection of Fig~\ref{fig:S_channel_combining}, we have
\begin{align}
\chi_{n}^{D} = \chi_{n-1}^{D} + \chi_{n-m}^{D}  + 2N_{n-m} \quad n,m \geq 1, \label{eqn:dec_complexity_recursion}
\end{align}
where $\chi_{0}^{D}=\chi_{-1}^{D}=\ldots =\chi_{1-m}^{D}=0$. 

The recursions in \eqref{eqn:enc_complexity_recursion} and \eqref{eqn:dec_complexity_recursion} are cumbersome to deal with. To observe the 
scaling behavior of $\chi_{n}^{E}$  and $\chi_{n}^{D}$ in $m$, we define  
\begin{align} \label{complex}
	 \eta^E \defn \frac{ \chi_{n}^{E} }{N \log N}, \quad \eta^D \defn \frac{ \chi_{n}^{D} }{N \log N},
\end{align}
and demonstrate the scaling of $\eta^E$ and $\eta^D$ in Fig~.\ref{fig:complexity}, where we have numerically calculated $\chi_{n}^{E}$ and $\chi_{n}^{D}$  as in \eqref{eqn:enc_complexity_recursion} and \eqref{eqn:dec_complexity_recursion} by choosing $N=O(\phi^n)$ to be the code-length closest to $10^4$ and $10^6$. From Fig.~\ref{fig:complexity} we observe that, there exist a decrease in $\eta_{n}^{E}$ and $\eta_{n}^{D}$ as  $m$ increases, where the decrease is more steep for small values of $m$ and it becomes more monotone as $m$ increases. This decrease in complexity, although not being orders of magnitude, is promising in showing the existence of polar codes requiring lower complexity. For example, from Fig.~\ref{fig:complexity} we observe that
$\eta_{n}^{D}$ is around $1/2$ when $m=12$. This indicates that the decoding complexity of $\{ {\mathscr C}_{n}^{(12)}\}$ is reduced by half compared to $\{ {\mathscr C}_{n}^{(1)} \}$  which is the polar code
presented by Ar\i kan in \cite{Arikan}. 

\begin{figure}[tb]
\begin{center}
\includegraphics[scale=0.4]{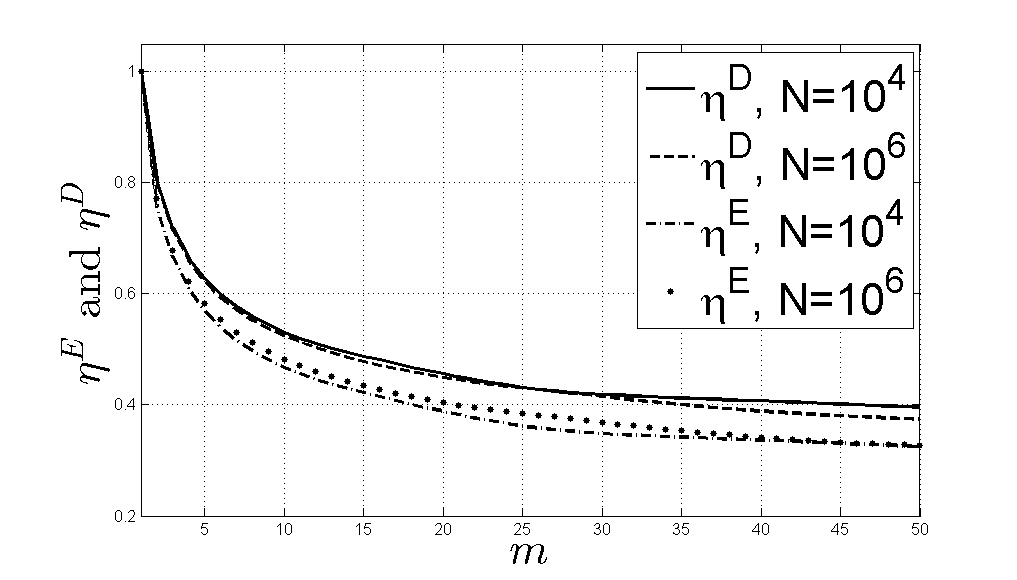}
\caption{Scaling of encoding and decoding complexities as $m$ increases where $N$ is chosen to be the code-length
closest to $1\times 10^4, 1\times 10^6$.}
\label{fig:complexity}
\end{center}
\end{figure}

\subsection{Sparsity}
As we have explained in Section \ref{sect:2}, there exist a sparsity in the channel combining process in the sense that at each combining level, the vector channel $W_n$ is obtained by combining $W_{n-1}$ and $\hat{W}_{n-m}$ which are obtained from $N(n-1)$ and $N(n-m)$ uses of underlying B-DMC, $W$, respectively. From Proposition \ref{prop:S_split_channels} we observe that the overall effect of channel combining and splitting is that, at each level $n$, there exist  $N(n-m)$ bit-channel pairs that participate in $\boxplus$ and $\boxminus$ transforms. As $m$ increases $N(n-m)$ decreases with respect to $N(n-1)$ implying the fraction of bit-channels participating in $\boxplus$ and $\boxminus$ transforms also decreases. On the other hand, as $m$ increases, the code-length increases less rapidly in $n$ because $N=O(\phi^n)$ and $\phi$ is decreasing in $m$, thus one can fit more channel combining and splitting levels within fixed code-length. A natural question is to understand the overall effect of increasing $m$ on the  total number of $\boxplus$ and $\boxminus$ transforms that one can obtain when the number of uses of $W$ channels is fixed.  The importance of $\chi_{n}^{D}$ in \eqref{eqn:dec_complexity_recursion} comes to play at this point because it gives us the total number of $\boxplus$ and $\boxminus$ transformation that are recursively applied to independent uses of $W$ channels to obtain the bit-channels in $W_n$. Consequently, one can view $\eta_{D}$ as a \textit{packing ratio} in the sense that one can pack $\eta_{n}^{D} {N \log N}$ recursive applications of $\boxplus$ and $\boxminus$ transformation to $N$ independent uses of $W$. Inspecting the scaling of $\eta_{D}$  in Fig.~\ref{fig:complexity} we observe that
this packing ratio is $1$ when $m=1$ and it decreases with increasing $m$, and this decrease manifests itself as a reduction in the decoding complexity of $\{ {\mathscr C}_{n}^{(m)}  \}$.

\section{Conclusion and Future Work}\label{sect:6}
We have introduced a method to design a class of code sequences $\{ {\mathscr C}_{n}^{(m)}; n \geq 1, m \geq 1\}$ with code-length
$N=O(\phi^n), \phi \in (1,2],$ and memory order $m$. The design of $\{ {\mathscr C}_{n}^{(m)} \}$ is based on the channel polarization idea of Ar\i kan \cite{Arikan} and $\{ {\mathscr C}_{n}^{(m)} \}$ coincides with the polar codes presented by Ar\i kan when $m=1$. We showed that 
$\{ {\mathscr C}_{n}^{(m)} \}$ achieves the symmetric capacity of arbitrary BDMCs for arbitrary but fixed $m$. We have obtained an achievable bound on the asymptotic polarization of performance of $\{ {\mathscr C}_{n}^{(m)}\}$ as scaled with $m$ and showed that the encoding and decoding complexities of $\{ {\mathscr C}_{n}^{(m)}\}$ decrease with increasing $m$. Our introduction of $\{ {\mathscr C}_{n}^{(m)}\}$ complements Ar\i kan's conjecture that channel polarization is a general phenomenon and it shows the existence of polar codes requiring lower complexity.  Future work will include a rate dependent analysis and a converse result on the asymptotic polarization performance of $\{ {\mathscr C}_{n}^{(m)}\}$.

\bibliographystyle{IEEEtran}
\bibliography{myref}

\section{Appendix}

\subsection{Proof of Proposition~\ref{prop:R_0}}
We have $J(W^-) = \frac{2}{1+Z(W^-)}$  and $J(W^+) = \frac{2}{1+Z(W^+)}$. By using \eqref{eqn:Z_minus} and \eqref{eqn:Z_plus} we obtain
\begin{align}
& J(W^+) + J(W^-)  \geq  \log \frac{2}{1+Z(W') Z(W'')} +  \notag \\
& \log \frac{2}{1+Z(W')+Z(W'')-Z(W') Z(W'')} \label{eqn:R_0_ineq_1} \\
& \! \! \!  = \log \frac{2}{1\! \! +\! \! Z(W') \! \! +\! \! Z(W'') + w(W',W'') Z(W') Z(W'')} \notag
\end{align}
where $w(W',W'')=Z(W')+Z(W'')-Z(W')Z(W'') \leq 1$ indicating
\begin{align}
 &\! \!  \! \!  \! \!  J(W^+) + \!  J(W^-)  \!  \geq  \!  \log \frac{2}{1+Z(W') } \!  + \!  \log \frac{2}{1+Z(W'')} \label{eqn:R_0_ineq_2}\\
& =J(W') + J(W'') \notag.  
\end{align}
In order to have $J(W^+) + J(W^-) = J(W') + J(W'')$, the equalities in  \eqref{eqn:R_0_ineq_1} and \eqref{eqn:R_0_ineq_2} must be achieved. From \eqref{eqn:Z_minus} we know that the equality in  \eqref{eqn:R_0_ineq_1}  is achieved only if $Z(W') \in \{0,1\}$ or $Z(W'') \in \{0,1\}$ or if $W'$ and $W''$ are BECs. When $(Z(W'),Z(W'')) \in (0,1)^2$ we have $w(W',W'') <1$ and the inequality in \eqref{eqn:R_0_ineq_2} is always strict, whether or not $W'$ and $W''$ being BECs. Consider the case $Z(W')=1$ or $Z(W'')=1$, then we have $w(W',W'')=1$ and the equalities  in \eqref{eqn:R_0_ineq_1} and \eqref{eqn:R_0_ineq_2} are achieved. When $Z(W')=0$ we have $J(W')=1$, $w(W',W'')=0$ and  $J(W^+) + J(W^-) = J(W')+J(W'')$, and the case $J(W')=1$ follows from the symmetry in \eqref{eqn:R_0_ineq_1} and \eqref{eqn:R_0_ineq_2}.
Hence the equalities in \eqref{eqn:R_0_ineq_1} and \eqref{eqn:R_0_ineq_2} are both achieved only if $Z(W') \in \{0,1\}$ or $Z(W'') \in \{0,1\}$, or
alternatively only if $J(W') \in \{0,1\}$ or $J(W'') \in \{0,1\}$.

\subsection{Proof of Proposition~\ref{prop:unique_states}}

From the operation of $\varphi_n$ in Defn.~\ref{defn:state_labeling} we obtain  ${\cal S}_1 = \{+,-\}$ such that $\textbf{s}_1^{(1)}=(+)$ and 
$\textbf{s}_1^{(2)}=(-)$, indicating $\textbf{s}_1^{(1)}$ and $\textbf{s}_1^{(2)}$ are unique. Proof is by induction, assume that
$s_{n-1}^{(j)} \in {\cal S}_{n-1}$ are unique. Let $j \in \mathbb{N}_{n-m}$ and consider $\textbf{s}_{n-1}^{(j)}$ to whom by appending $+$ and $-$ one obtains $\textbf{s}_{n}^{(j)}$ and  $\textbf{s}_{n}^{(j+N(n-1))}$, respectively, indicating $\textbf{s}_{n}^{(j+N(n-1))}$ and $\textbf{s}_{n}^{(j)}$ are different from each other. Next, let $j \in \mathbb{N}_{n-1}\setminus \mathbb{N}_{n-m} $ then $\textbf{s}_{n}^{(j)}$ are obtained by appending $\bigstar$ to $\textbf{s}_{n-1}^{(j)}$ which, by assumption, are unique. Combining the result we see that for all 
$j \in  \mathbb{N}_{n}$  the vectors $s_n^{(j)} \in {\cal S}_n$ are different from each other.

\subsection{Proof of Proposition~\ref{prop:same_state_history} }

Investigating Fig~\ref{fig:S_state_labeling} consider the operation of $\varphi_{n-1}$ where 
$s_{n-2}^{(k)}=(s_1,s_2,\ldots,s_{n-2})$, $k \in \mathbb{N}_{n-2}$, holds at level $n-1$. Next, consider the operation of $\varphi_{n-2}$ where 
one has $s_{n-3}^{(k)}=(s_1,s_2,\ldots,s_{n-3})$ for $k \in \mathbb{N}_{n-3}$. In turn and by induction through $\varphi_{n-2}, \varphi_{n-3},\ldots,\varphi_{n-(m-1)}$ we conclude that $s_{n-m}^{(j)}=(s_1,s_2,\ldots,s_{n-m})$, $j \in \mathbb{N}_{n-m}$.

\subsection{Proof of Proposition~\ref{prop:Dominating_root}}
\textit{i)} For $m>1$ we have $F(m,1)=-1<0$ and $F(m,2)=2^{m-1}-1 \geq 0$ so that there exists at least one real root in (1,2].  Proof is by contradiction, let $\rho_1, \rho_2 \in (1,2]$ be two real roots of $F(m,\rho)$ then from \eqref{eqn:roots_of_N} we have
\begin{align}
	{\rho}_1^{m-1}(\rho_1-1)=1 \label{eqn:real_root_1}, \\
	{\rho}_2^{m-1}(\rho_2-1)=1 \label{eqn:real_root_2}.
\end{align}
Let $\rho_1< \rho_2$, then $\rho_2^{m-1}>{\rho}_1^{m-1}$ and 
$\rho_2-1>\rho_1-1>0$ implying ${\rho}_2^{m-1}({\rho}_2-1)>1$ if ${\rho}_1^{m-1}(\rho_1-1)=1$ which contradicts  $\eqref{eqn:real_root_2}$, carrying a similar analysis for $\rho_1< \rho_2$ also contradicts $\eqref{eqn:real_root_2}$, which indicates $\rho_1=\rho_2=\phi$. 

\textit{ii)}   Assume that $\rho$ is a complex root of $F(m,\rho)$, with $\sqrt{\rho \rho^*}=\sigma> 1$ where $*$ denotes the conjugate operation. Since the coefficients of $F(m,\rho)$ are real, its complex roots must be in conjugate pairs. From \eqref{eqn:roots_of_N}
\begin{align*}
\rho^{m-1}(\rho-1)=1, \\
{\rho^*}^{m-1}(\rho^*-1)=1. 
\end{align*}
Multiplying the above equations we obtain
\begin{align}
\sigma ^{2(m-1)}(\sigma^2 -2Re(\rho) +1 )=1, \notag \\
\sigma ^{2(m-1)}(\sigma^2 -2\sigma \alpha +1 )=1, \label{eqn:root_eq_2}
\end{align}
where $0 \leq \alpha <1$. In turn for any $\rho$, $\sigma$ must be a root of 
\begin{align}\label{eqn:complex_root_ineq}
g(\sigma,\alpha)=\sigma ^{2(m-1)}(\sigma^2 -2\sigma \alpha +1 )-1,
\end{align}
Observe that when $\sigma$ is fixed $g(\sigma,\alpha)$ is decreasing in $\alpha$. We also have
\begin{align*}
\begin{split}
 \frac {\partial g(\sigma,\alpha)}{\partial \sigma} &={2(m-1)}\sigma^{2(m-1)-1}(\sigma^2 -2\sigma \alpha +1 ) \\ 
&\quad +\sigma^{2(m-1)}(2\sigma-2\alpha) 
\end{split}
\end{align*} 
From  \eqref{eqn:root_eq_2} observe that $(\sigma^2 -2\sigma \alpha +1 ) >0$, and since $(2\sigma -2\alpha)>0$ for
$\sigma > 1$ we have $ \frac {\partial g(\sigma,\alpha)}{\partial \sigma} >0$. This indicates that $g(\sigma,\alpha)$ is increasing with $\sigma$. But $\phi$ is a root of $g(\sigma,\alpha)$ with $\alpha=1$ and thus $g(\phi,1)=0$. Since $g(\sigma,\alpha)$ is decreasing in $\alpha$ we have $g(\phi,\alpha) \geq 0$  and $g(\sigma,\alpha) = 0$ is only achieved if $\sigma < \phi$ because $g(\sigma,\alpha)$ is increasing with $\sigma$.

\textit{iii)} Observe that for some $\rho \in (1,2]$ we have $ \frac{\partial F(m,\rho)}{\partial \rho }>0$ so that $F(m,\rho)$ is increasing in $\rho$ and when $\rho$ is fixed $F(m,\rho)$ is also increasing in $m$. Assume that $ \rho_1, \rho_2 \in (1,2]$ are real roots of $F(m_1,\rho)$ and $F(m_2,\rho)$, respectively, where $m_1,m_2\geq 1$. Then $f(m_1,\rho_1) < f(m_2,\rho_1)$ holds if $m_2 > m_1$ and $ f(m_1,\rho_1)  = f(m_2,\rho_2)=0$ is satisfied only if $\rho_1 < \rho_2$.

\subsection{Proof of Lemma~\ref{lemma:sum_cut_off_eq}}

Let $J_n^{(i)} =J (W_n^{(i)})$ denote symmetric cut-off rate of $W_n^{(i)}$. From Proposition~\ref{prop:S_split_channels} we know that
for $j \in \mathbb{N}_{n-m}$ we have $W_n^{(j)}=W_{n-1}^{(j)} \boxplus W_{n-m}^{(j)}$ and $W_n^{(j+N(n-1))}=W_{n-1}^{(j)} \boxminus W_{n-m}^{(j)}$. Proposition~\ref{prop:R_0} indicates that these transforms increase the sum cut-off rate as $J_n^{(j)}+J_n^{(j+N(n-1))} \geq J_{n-1}^{(j)}+J_{n-1}^{(j)}$ where the equality is achieved only if $J_{n-1}^{(j)} \in \{0,1\}$ or $J_{n-m}^{(j)} \in \{0,1\}$ holds. For $j \in \mathbb{N}_{n-1} \setminus \mathbb{N}_{n-m}$, from  Proposition~\ref{prop:S_split_channels}, we have $J_{n}^{(j)}=\gamma(n) J_{n-1}^{(j)}$
which implies $J_n^{(j)}=J_{n-1}^{(j)}$. Combining the above results gives
\begin{align*}
\sum_{i \in \mathbb{N}_n} J_n^{(i)} \geq \sum_{j \in \mathbb{N}_{n-1} } J_n^{(j)} + \sum_{k \in \mathbb{N}_{n-m}} J_n^{(k)},  
\end{align*}
where the equality is achieved only of  if $J_{n-1}^{(j)} \in \{0,1\}$ or $J_{n-m}^{(j)} \in \{0,1\}$ holds for all $j \in \mathbb{N}_{n-m}$. In the probabilistic domain of Section\ref{sect:4} the above result is equivalent to
\begin{align*}
\sum_{\textbf{s}_ n \in {\cal S}_n  } J_n \geq \sum_{\textbf{s}_ {n-1} \in {\cal S}_{n-1}  } J_{n-1} + \sum_{\textbf{s}_ {n-m} \in {\cal S}_{n-m}  } J_{n-m},  
\end{align*}
where the equality is achieved only of  if $J_{n-1} \in \{0,1\}$ or $J_{n-m} \in \{0,1\}$ holds for all $S_n  \in \{+,-\}$. Dividing both sides of the above inequality by $1/N(n)$  and using $ E[J_n]=\frac{1}{N(n)} \sum_{\textbf{s}_ n \in {\cal S}_n  } J_n$ we obtain 
\begin{align*}
E[J_n] \geq \frac{N(n-1)}{N(n)} E[J_{n-1}] + \frac{N(n-m)}{N(n)} E[J_{n-m}]. 
\end{align*}
Noticing $\frac{N(n-1)}{N(n)}=\mu(n)$ and $ \frac{N(n-m)}{N(n)}= 1-\mu(n)$ completes the proof.

\subsection{Proof of Lemma~\ref{prop:min_symmetric_cut_off}}
From \eqref{eqn:j_sum_1} we have
\begin{align}
 E[J_n]  \geq \mu E[J_{n-1}] + \left(1-\mu \right) E[J_{n-m}], \notag \\
              \geq \min \{ E[J_{n-1}] , E[J_{n-m}] \} \label{eqn:cutoff_min_ineq_2},
\end{align}
Let us define the set
\begin{align*}
{\cal E}_k^{(m)}\defn \{E_{km},E_{km-1},\ldots, E_{km-(m-1)}\}. 
\end{align*}
By definition in \eqref{eqn:decimated_process} we have we have $E[\hat{J}_k] = \min {\cal E}_k^{(m)}$.
Proof is by induction.  We use \eqref{eqn:cutoff_min_ineq_2} to upper bound the elements of  ${\cal E}_{k}^{(m)}$ with respect to 
$\min {\cal E}_{k-1}^{(m)}=E[\hat{J}_{k-1}]$. Let $n=km-(m-1)$ and use \eqref{eqn:cutoff_min_ineq_2} to obtain
\begin{align*} 
E_{km-(m-1)} &\geq \min \{ E_{(k-1)m},E_{(k-1)m-(m-1) }  \}, \\
		              &\geq \min {\cal E}_{k-1}^{(m)} 
\end{align*}
For $i=2,3,\ldots,m-1$ assume 
\begin{align*}
E_{km-(m-i)} \geq \min {\cal E}_{k-1}^{(m)}
\end{align*}
holds. Next, let $n=km-(m-(i+1))$ in \eqref{eqn:cutoff_min_ineq_2} to write
\begin{align*} 
E_{km-(m-(i+1)) } &\geq \min \{  E_{km-(m-i)} , E_{(k-1)m-(m-(i+1)) }  \}.
\end{align*}
By assumption $E_{km-(m-i)} \geq  \min {\cal E}_{k-1}^{(m)}$ and by definition $E_{(k-1)m-(m-(i+1))} \geq \min {\cal E}_{k-1}^{(m)}$ holds, 
indicating
\begin{align*} 
E_{km-(m-(i+1)) } &\geq\min {\cal E}_{k-1}^{(m)}.
\end{align*}
Combining the above results tells us for $i=1,2,\ldots,m$ we have $E_{km-(m-i)} \geq  \min {\cal E}_{k-1}^{(m)}=E[\hat{J}_{k-1}]$ which
indicates $E[\hat{J}_{k}] \geq E[\hat{J}_{k-1}]$.

\subsection{Proof of Lemma~\ref{lemma:class_size_upper_bound}}

In order to bound $| {\cal T}_n^{(q)}|$ we decompose ${\cal T}_n^{(q)}$ it into two different sets
\begin{align*} 
{\cal T}_n^{(a,q)}& \defn \left\{ \textbf{s}^n : P_{\textbf{s}^n}^{(\minus)}=q, s_n=+  \right\} , \\
{\cal T}_n^{(b,q)}& \defn \left\{ \textbf{s}^n : P_{\textbf{s}^n}^{(\minus)}=q, s_n \neq + \right\}
\end{align*}
and we have $T_n^{(q)}={\cal T}_n^{(a,q)} \cup {\cal T}_n^{(b,q)}$.
Recall that each state $-$ in $\textbf{s}_n$ is followed by $m-1$ occurrences of state $\bigstar$. In turn, 
${\cal T}_n^{(a,q)}$ consists of $\textbf{s}_n$ having $k=nq$, $0 \leq k \leq n/m$, occurrences of the vector 
$\textbf{a}=(-,\underbrace{\bigstar,\bigstar,\ldots,\bigstar}_{m-1 \text{ times}})$ and $n-km$ occurrences of state $+$. By combinatorial analysis we have
\begin{align*}
|{\cal T}_n^{(a,q)}| =  { n-(m-1)k \choose k }.
\end{align*}
${\cal T}_n^{(b,q)}$ consists of $k-1$ occurrences of the vector $\textbf{a}$, 
an occurrence of  $\textbf{b}=(-,\underbrace{0,0,\ldots,0}_{p  \text{ times}})$, $1 \leq  p< m-1$, and $n-mk-(p+1)$ occurrences of state $+$. 
The vector $\textbf{b}$ can only occur in the last $p+1$ entries  in $\textbf{s}_n$ and  it will be completed to a vector $\textbf{a}$  
if we had prolonged the channel combining operation $m-1-p \leq m$ more levels. Therefore
\begin{align*}
|{\cal T}_n^{(b,q)}| & \leq  { n+m-(m-1)k \choose k }.
\end{align*}
For some $c \in \mathbb{Z}$ and $d \in \mathbb{Z}$ with $c<d$ we have ${d \choose c }= \frac{d}{d-c} {d-1 \choose c } \leq d{d-1 \choose c }$, using this fact we obtain
\begin{align*}
     { n+m-(m-1)k \choose k } & \leq (n+m){ n+(m-1)-(m-1)k \choose k },\\
	      & < (n+m)^2{ n+(m-2)-(m-1)k \choose k } \\
	       &  \quad \quad \vdots		\\
	      & < (n+m)^m{ n-(m-1)k \choose k }	
\end{align*}
Then we have
\begin{align}
|{\cal T}_n^{(q)}| &=  |{\cal T}_n^{(a,q)}|+|{\cal T}_n^{(b,q)}|, \notag\\
	   &<\left( 1+(n+m)^m \right) { n-(m-1)k \choose k }, \notag	\\
           & <\left( 1+(n+m) \right)^m { n-(m-1)k \choose k },	\notag \\
              & =2^{nB(m,n)}{ n-(m-1)k \choose k } \label{eqn:type_class_upper_1},	
\end{align} 
where $B(m,n)=\frac{m \log(1+n+m)}{n}=o(1)$. Next, we use the upper bound ${ n \choose k} \leq 2^{nH(k/n)}$ in \cite{Cover} to upper bound ${ n-(m-1)k \choose k }$  as
\begin{align}
{ n-(m-1)k \choose k } &\leq 2^{n(1-(m-1)(k/n))H(\frac{(k/n)}{1-(m-1)(k/n)}}, \notag  \\
                                     &= 2^{nG(m,q)} \label{eqn:type_class_upper_2}. 
\end{align}
Combining \eqref{eqn:type_class_upper_1} and  \eqref{eqn:type_class_upper_2}we obtain the desired bound as $|T_n^{(q)}| < 2^{n(G(m,q)+B(m,n))}=2^{n(G(m,q)+o(1))}$.

\subsection{Proof of Lemma~\ref{lemma:root_and_bound_duality}}
We have
\begin{align*}
G(m,q)= (1-(m-1)q)H \left( \frac{(q)}{1-(m-1)q} \right).
\end{align*}
We know that, for $q \in [0,1/m]$, $H(\frac{q}{1-(m-1)q})$ is concave in $q$ and $(1-(m-1)q)$ is linear in $q$ indicating $G(m,q)$ is concave in $q$. Let $q*$ denote the maximizer of $G(m,q)$. The maximum of $H(\frac{q}{1-(m-1)q})$ occurs when $\frac{q}{1-(m-1)q}=\frac{1}{2}$ or equivalently when $q=\frac{1}{m+1}$ and since $(1-(m-1)q)$ is decreasing in $q$, we have  $q* \in[0,\frac{1}{m+1}]$. We next evaluate 
$\frac{\partial G(m,q)}{\partial q}$ 
\begin{align*}
\begin{split}
\frac{\partial G(m,q)}{\partial q} &= (m-1)\log (1-(m-1)q)  \\
&  \quad \quad +\log q - m \log (1-mq). 
\end{split}
\end{align*}
setting  $\frac{\partial G(m,q)}{\partial q}|_{q=q*}=0$ gives
\begin{align}
\! \! \!\!\!\! (m-1) \log (1-(m-1)q^*) +\!\! \log q^* = m \log (1-mq^*). \label{eqn:q_star}
\end{align}
Re-arranging the above equation we obtain
\begin{align}\label{eqn:prob_log_eq}
\begin{split}
&m \log \frac{(1-(m-1)q^*)}{1-mq^*} + \log \frac{q^*}{1-mq^*} \\
& \quad  = \log\frac{ (1-(m-1)q^*)}{1-mq^*}.
\end{split}
\end{align}
Let us use the following substitutions
\begin{align*}
\eta = \frac{1-(m-1)q^*}{1-mq^*}, \quad \eta-1 = \frac{q^*}{1-mq^*}.
\end{align*}
For $q* \in [0,\frac{1}{m+1}]$ we have $\eta \in [1,2]$. Using the above substitutions in \eqref{eqn:prob_log_eq} we obtain
\begin{align*}
m \log \eta+\log (\eta-1) = \log \eta,
\end{align*}
or alternatively
\begin{align*}
\eta^m(\eta-1) =\eta.
\end{align*}
Dividing both sides of the above relation by $\eta$ and re-arranging the terms we obtain
\begin{align}
\eta^m-\eta^{m-1}-1 =0.
\end{align}
But the above polynomial is same as \ref{eqn:roots_of_N}. Consequently from part \textit{i} of Proposition.~\ref{prop:Dominating_root} we conclude that  $\eta = \phi$ which indicates that $\frac{1-(m-1)q^*}{1-mq^*}=\phi$ and hence  $q^* =\frac{1}{1+m(\phi-1)}=p^-$.
Next we evaluate the maximum of $G(m,q)$ attained at $q=q^*$. 
\begin{align}
G(m,q^*) = -q^*\log \frac{q^*}{1-(m-1)q^*} +  \notag \\
(mq^*-1) \log \frac{1-mq^*}{1-(m-1)q^*} \label{eqn:G_m_1}
\quad \end{align}
Re-arranging \eqref{eqn:q_star} we observe that
\begin{align*}
\log \frac{q^*}{1-(m-1)q^*}=m\log \frac{1-mq^*}{1-(m-1)q^*}
\end{align*}
Using the above relation in \eqref{eqn:G_m_1} gives
\begin{align*}
G(m,q^*) = \log \frac{1-(m-1)q^*}{1-mq^*}=\log \phi.
\end{align*}

\subsection{Proof of Proposition~\ref{prop:typical_set_S} }
 We define a typical set ${\cal T}_n^{(q,\epsilon)}$ as
\begin{align*}
{\cal T}_n^{(q,\epsilon)} =\{ \textbf{s}_n : P_{s^n}^{(-)}= q, D(q,p^-) \leq \epsilon \}.
\end{align*}
The probability that ${\cal T}_n^{(q)}$ is not typical is 
\begin{align}
  1-\Pr( {\cal T}_n^{(q,\epsilon)})&=  \sum_{\Pr ( D(q,p^-)  >\epsilon )} \Pr ({\cal T}_n^{(q)}), \nonumber  \\
      & \stackrel{a}{\leq}   \sum_{\Pr(D(q,p^-)  >\epsilon)}    2^{-n\left(D(q,s_-)  +o(1) \right) } , \nonumber \\
      & \stackrel{}{\leq}    \sum_{\Pr (D(q,p^-)  >\epsilon)}    2^{-n\left(\epsilon +  o(1) \right) }, \nonumber \\
      &  \stackrel{b}{\leq}(n+1) 2^{-n\left(\epsilon+o(1)\right) }, \nonumber \\
      & = 2^{-n\left(\epsilon+o(1)\right)}  \label{eqn:typicality},
\end{align}
In the above derivation (a) follows from \eqref{eqn:pTk} and (b) follows from the fact that there exist at most $n+1$ different type classes having $\Pr (D(q,s_-)>\epsilon)$.
The above result indicates that $\sum_{n \rightarrow \infty} \Pr( D(q,s_-) \geq \epsilon ) $ converges, thus the expected number of the occurrences of the event  $D(q,s_-) >\epsilon$ for all $n$ is finite. By using the first Borel Cantelli Lemma \cite[p. 59]{Billingsley} we conclude that $ D(q,s_-)$ converges to $0$ with probability 1.

\subsection{Proof of Lemma~\ref{lemma:direct_bound} }

Conditioned on the event $D_{n_0}^n(\gamma) = \#\big( (s_{n_0+1},\ldots,s_{n})|+\big) \geq \gamma(n-n_0) $  there exists at least $\gamma(n-n_0)$ occurrences of state $+$ in
$\{ S_{n_0+1}, S_{n_0+2}, \ldots, S_{n} \}$. Investigating  \eqref{eqn:Bhat_Process_1}, we have $\hat{Z}_n \leq \hat{Z}_{n-1}$ when $S_n=+$ and $Z_n\geq Z_{n-1}$
when $S_n \neq +$. Moreover, $Z_n$ is increasing in $Z_{n-1}$ when $S_n$ is fixed. Consequently, if we fix $\hat{Z}_m$,
the largest value of $\hat{Z}_n$ will occur if $\{ S_{n_0+1}, S_{n_0+2}, \ldots\,S_{n} \}$ has the following realization
\begin{align*}
\{ \!\!\!\!\!\!\! \overbrace{\textbf{a},\textbf{a},\ldots,\textbf{a},}^{(1-\gamma)(n-n_0)/m \text{ times}}\!\!\!\!\!\!\!\!\underbrace{+,+,\ldots,+}_{\gamma(n-n_0) \text{ times}}\}.
\end{align*}
where $\textbf{a}=(-,\underbrace{\bigstar,\bigstar,\ldots,\bigstar}_{m-1 \text{ times}})$. In order to upper bound $\hat{Z_n}$ we assume that
the above realization has occured for $\{ S_{n_0+1}, S_{n_0+2}, \ldots\,S_{n} \}$. During consecutive runs of $+$, the value of $\log \hat{Z}_n$ increases with the same recursion as the code-length in \eqref{eqn:code_length_recursion} as $\log \hat{Z}_{n}= \log \hat{Z}_{n-1}+ \log \hat{Z}_{n-m}$. This recursion happens $\gamma(n-m)$ times and since the code-legth obeying the same recursion scales as $\phi^{\gamma(n-m)}$, $\phi \in (1,2]$, we have 
\begin{align}
\log \hat{Z}_n = \phi^{\gamma(n-n_0)}\log \hat{Z}_k, \label{eqn:inter_1}
\end{align}
where $k=n_0+(1-\gamma)(n-m)$. During consecutive runs of $\textbf{a}$ the value of $\hat{Z}_i$ does not change with respect to $\hat{Z}_{i-1}$ when $S_i = \bigstar$ and it increases as $\hat{Z}_i=\hat{Z}_{i-1}+\hat{Z}_{i-m}-\hat{Z}_{n-1}\hat{Z}_{i-m}$ when $S_i=-$.
By construction of $\{ S_{n_0+1}, S_{n_0+2}, \ldots\,S_{n} \}$ each state $-$ is preceed by $m-1$ occurances of $\bigstar$ therefore if $S_i=-$
we have $(S_{i-1},S_{i-2},\ldots,S_{i-(m-1)})=(\bigstar, \bigstar,\ldots,\bigstar)$ indicating $\hat{Z}_{i-1}=\hat{Z}_{i-2}=\ldots=\hat{Z}_{i-(m-1)}$.
Therefore during each occurance of state $-$ in $\textbf{a}$ we see the recursion $\hat{Z}_{i-1}+\hat{Z}_{i-m}-\hat{Z}_{i-1}\hat{Z}_{i-m}=2\hat{Z}_{i-1}-\hat{Z}_i^{(i)}$ or equvalently $1-\hat Z_{i}=(1-\hat{Z}_i^{(i)})^2$. This recursion occurs $(1-\gamma)(n-n_0)$ times resulting
in $1-\hat{Z}_k=(1-\hat Z_{n_0})^{2(1-\gamma)(n-n_0)}$ and $\hat{Z}_k=1-(1-\hat Z_{n_0})^{2(1-\gamma)(n-n_0)}$.  Next, employ the inequality $\log x \leq x-1$, $x \in [0,1]$, by letting $x =\hat{Z}_k$ to obtain
\begin{align}
\log  \hat{Z}_k  \leq -(1-\hat Z_{n_0})^{2(1-\gamma)(n-n_0)} . \label{eqn:inter_2}
\end{align}
Using \eqref{eqn:inter_2} in \eqref{eqn:inter_1} gives
\begin{align*}
\log  \hat{Z}_n &= - \phi^{\gamma(n-n_0)}(1-Z_{n_0})^{2(1-\gamma)(n-n_0)}, \\
& \leq  - \phi^{\gamma(n-n_0)}(1-Z_{n_0})^{2(n-n_0)} \\
&=-\phi^{(\gamma-\epsilon)(n-n_0)}\left( (1-Z_{n_0})^{2}\phi^{ \epsilon} \right)^{(n-n_0)}.
\end{align*}
Choose $\zeta \in (0,1)$ so that $\zeta \leq 1-\phi^{\frac{-\epsilon}{2}}$ holds. Conditioned on $C_{n_0}(\zeta)=\{ Z_{n_0} \leq \zeta \}$ we have
$(1-Z_{n_0})^{2}\phi^{ \epsilon} \geq 1$, resulting in
\begin{align*}
\log_2 \hat{Z}_n \leq -\phi^{(\gamma-\epsilon)(n-m)}, \quad C_{n_0}(\zeta) \cap D_{n_0}^n(\gamma),
\end{align*}
which proves the lemma.

%
%


%




\end{document}